%% file: Draft_v4.tex
\documentclass[11pt, letter]{article}

\usepackage[T1]{fontenc}    
\usepackage{hyperref}       
\usepackage{url}            
\usepackage{booktabs}       
\usepackage{amsfonts}       
\usepackage{nicefrac}       
\usepackage{microtype}      

\usepackage[compact]{titlesec}
\usepackage{enumitem}
\usepackage{wrapfig}

\usepackage{amsmath,amstext,amsopn,amsthm}
\usepackage{mathtools}
\usepackage{cleveref}
\usepackage{graphicx}
\usepackage{color}

\usepackage{bbm}
\usepackage{algorithm}

\usepackage[noend]{algpseudocode}
\newtheorem{theorem}{Theorem}
\newtheorem{corollary}{Corollary}
\newtheorem{lemma}{Lemma}

\newtheorem{fact}{Fact}

\newtheorem{observation}{Observation}

\makeatletter
\def\BState{\State\hskip-\ALG@thistlm}
\makeatother

\def\RR{{\mathbb R}}

\def\HH{{\mathbb H}}

\def\ZZ{{\mathbb Z}}

\def\P{{\mathcal P}}
\def\EE{{\mathbb E}}

\def\K{{\mathcal K}}

\def\M{{\mathcal M}}
\def\A{{\mathcal A}}
\def\I{{\mathcal I}}
\def\D{{\mathcal D}}
\def\Qe{{\mathcal Q}}

\DeclareMathOperator{\argmax}{argmax}

\DeclareMathOperator{\OPT}{OPT}
\DeclareMathOperator{\alg}{ALG}

\newcommand{\one}{\mathbf{1}}
\DeclareMathOperator{\ident}{\bf I}
\DeclareMathOperator{\SigmaB}{\bf \Sigma}
\DeclareMathOperator{\muB}{\bf \mu}
\DeclareMathOperator{\ub}{UB}
\DeclareMathOperator{\lb}{LB}

\usepackage{mathtools}
\renewcommand{\algorithmicrequire}{\textbf{Input:}}
\renewcommand{\algorithmicensure}{\textbf{Output:}}

\usepackage{sidecap}

\usepackage[suppress]{color-edits}
\addauthor{nh}{magenta}
\addauthor{at}{blue}

\usepackage[margin=1.2in]{geometry}

\usepackage{thmtools}
\usepackage{thm-restate}

\usepackage{hyperref} 
\usepackage[round]{natbib}
\title{Structured Robust Submodular Maximization:  \\ Offline and Online
  Algorithms}
\usepackage{times}

\author{Alfredo Torrico\thanks{Polytechnique Montr\'eal, Canada. Email: alfredo.torrico-palacios@polytml.ca} \quad Mohit Singh\thanks{Georgia Institute of Technology, Atlanta. Email: mohit.singh@isye.gatech.edu}  \quad Sebastian Pokutta\thanks{ZIB, TU Berlin. Email: pokutta@zib.de} \\ Nika Haghtalab\thanks{Cornell University, Email: nika@cs.cornell.edu}  \quad Joseph (Seffi) Naor\thanks{Technion, Haifa, Israel. Email: naor@cs.technion.ac.il} \quad Nima Anari\thanks{Stanford University, Stanford. Email: anari@berkeley.edu} }

\date{}

\begin{document}

\maketitle

\begin{abstract}
Constrained submodular function maximization has been used in subset selection problems such as selection of most informative sensor locations. While these models have been quite popular, the solutions Constrained submodular function maximization has been used in subset selection problems such as selection of most informative sensor locations. While these models have been quite popular, the solutions obtained via this approach are unstable to perturbations in data defining the submodular functions. Robust submodular maximization has been proposed as a richer model that aims to overcome this discrepancy as well as increase the modeling scope of submodular optimization.

In this work, we consider robust submodular maximization with structured combinatorial constraints and give efficient algorithms with provable guarantees. Our approach is applicable to constraints defined by single or multiple matroids, knapsack as well as distributionally robust criteria. We consider both the offline setting where the data defining the problem is known in advance
as well as the online setting where the input data is revealed over time. For the offline setting, we give a general (nearly) optimal bi-criteria approximation algorithm that relies on new extensions of classical algorithms for submodular maximization. For the online version of the problem, we give an algorithm that returns a bi-criteria solution with sub-linear regret.
 
\end{abstract}


\input{intro}
\input{formulation}

\input{greedy}

\input{monotonicity}

\input{offline_exp}

\input{online}

\input{extensions}

\subsection*{Acknowledgements}
This research was partially supported by NSF Award CCF-1717947, NSF
CAREER Award CMMI-1452463, an IBM Ph.D. fellowship, and a Microsoft
Research Ph.D. fellowship. We also thank Shabbir Ahmed for the discussions about the distributionally robust problem \eqref{eq:dist_robust}.

\bibliography{bibliography}

\appendix

 \input{appendix}

\end{document}

%% file: intro.tex
\section{Introduction}
Constrained submodular function maximization has seen significant progress in recent years in the design and analysis of new algorithms with guarantees~\citep{calinescu2011maximizing,ene_etal16, buchbinder_etal16, sviridenko_04}, as well as numerous applications - especially in constrained subset selection problems~\citep{powers_etal16, lin_etal09, krause_etal05, krause_etal09, krause2008robust, krause_etal08a} and more broadly machine learning. A typical example is the problem of picking a subset of candidate sensor locations for spatial monitoring of certain phenomena such as temperature, ph values, humidity, etc. \citep{krause2008robust}.  Here the goal is typically to find sensor locations that achieve the most coverage or give the most information about the observed phenomena. Submodularity naturally captures the decreasing marginal gain in the coverage, or the information acquired about relevant phenomena by using more sensors \citep{das2008algorithms}. While submodular optimization offers an attractive model for such scenarios, there are a few key shortcomings, which motivated \emph{robust submodular optimization} in the cardinality case \citep{krause2008robust}, so as to optimize against several functions \emph{simultaneously}:

\begin{enumerate}[topsep=0pt,itemsep=0ex,partopsep=1ex,parsep=1ex]
\item The sensors are typically used to measure various parameters at the same time. Observations for these parameters need to be modeled via different submodular functions. 

\item Many of the phenomena being observed are non-stationary and highly variable in certain locations. To obtain a good solution, a common approach is to use different submodular functions to model different spatial regions. 

\item The submodular functions are typically defined using data obtained from observations, and imprecise information can lead to unstable optimization problems. Thus, there is a desire to compute \emph{solutions that are robust} to perturbations of the submodular functions.
\end{enumerate}

Given the computational complexity of optimizing several functions simultaneously, \citet{krause2008robust} motivated a bicriteria approach. In simple words, in order to obtain provable guarantees, one needs to trade off the quality of the solution measured by its objective value with the ``size'' of the solution. In the case of a single cardinality constraint, \citet{krause2008robust} proposes a natural relaxation which consists in allowing more elements in the final set, i.e., violating the feasibility constraint. However, to obtain provable guarantees for more general \emph{combinatorial constraints} relaxing the size of feasible sets is not enough.

Our main contribution is the development of new algorithms with provable guarantees for robust submodular optimization under a large class of combinatorial constraints. These include partition constraints, where local cardinality constraints are placed on disjoint parts of the ground set.  More generally, we consider matroid and knapsack constraints.

We provide bi-criteria approximations that trade-off the approximation factor with the ``size'' of the solution, measured by the number \(\ell\) of feasible sets \(\{S_i\}_{i \in [\ell]}\) \atdelete{in the union \(S = \bigcup_{i \in [\ell]} S_i\) that} \atedit{whose union} constitutes the final solution \(S\). While this might be nonintuitive at first, it turns out that the union of feasible sets corresponds to an appropriate analog of the single cardinality constraint\atedit{.} \atdelete{case for} \atedit{Some} special cases of interest \atedit{are}:

\begin{enumerate}[topsep=0pt,itemsep=0ex,partopsep=1ex,parsep=1ex]
\item \emph{Partition constraints.} Given a partition of the candidate sensor locations, \atedit{the feasible sets correspond to subsets that satisfy a cardinality constraint on each part of the partition.} \atdelete{the partition constraints introduces a \emph{set} of cardinality constraints; one for each part of the partition.} The union of feasible sets here corresponds to relaxing the cardinality constraints \atedit{separately for each part}. This results in a stronger guarantee than relaxing the constraint globally as would be the case in the single cardinality constraint case.
\item \emph{Gammoid.} Given a directed graph and a subset of nodes $T$, the feasible sets correspond to subsets $S$ that can reach $T$ via disjoint paths in the graph. Gammoids appear in flow based models, for example in reliable routing. The union of feasible sets now corresponds to sets $S$ that can reach $T$ via paths such that each vertex appears in few paths. 
\end{enumerate}

We consider both offline and online versions of the problem, where the data is either known a-priori or is revealed over time, respectively. For the offline version of the problem, we provide a general procedure that iteratively utilizes any standard algorithm for submodular maximization to produce a solution which is a union of multiple feasible sets. The analysis relies on known insights about the performance of the classical greedy algorithm when is used for cardinality constraint. For the online case, we introduce new technical ingredients that might be broadly applicable in online robust optimization. Our work significantly expands on previous works on robust submodular optimization that focused on a single cardinality constraint~\citep{krause2008robust}. Moreover, our work substantially differ from its proceeding version \citep{torrico_etal19}, because: (1) we provide a general and flexible framework to design an offline bi-criteria algorithm; (2) we give significant implementation improvements and support our theoretical results via an exhaustive computational study on two real-world applications.

%% file: formulation.tex
\subsection{Problem Formulation}\label{sec:problem_formulation}

\nhedit{As we describe below, we study offline and online variations of \emph{robust submodular maximization under structured combinatorial constraints}. While our results holds for more general constraints, we focus our attention first on \emph{matroid} constraints that generalize the partition as well as the gammoid structural constraints mentioned above. We discuss extensions to other class of constraints in Section~\ref{sec:extensions}.}

\nhedit{Consider a non-negative set function $f:2^V\to\RR_+$. We denote the marginal value for any subset $A\subseteq V$ and $e\in V$ by $f_A(e):=f(A+e)-f(A)$, where $A+e:= A\cup\{e\}$. Function $f$ is {\it submodular} if and only if it satisfies the {\it diminishing returns property}. Namely, for any $e\in V$ and $A\subseteq B\subseteq V\backslash\{e\}$, $f_A(e)\geq f_B(e)$.
We say that $f$ is \emph{monotone} if for any $A\subseteq B\subseteq V$, we have $f(A)\leq f(B)$.
Most of our results are concerned with optimization of monotone submodular functions.}

A natural class of constraints considered in submodular optimization are \emph{matroid} constraints.
For a ground set $V$ and a family of sets $\I\subseteq 2^V$, $\M=(V,\I)$ is a matroid if
(1) for all $A\subset B\subseteq V$, if $B\in \I$ then $A\in \I$ and
(2) for all $A, B\in \I$ with $|A| < |B|$, there is $e\in B\setminus A$ such that $A \cup \{e\} \in \I$.
Sets in such a family $\I$ are called \emph{independent} sets, or simply put, \emph{feasible} sets for the purpose of  optimization. Maximal independent sets are called \emph{basis}. Finally, the \emph{rank} of a matroid is the maximum size of an independent set in the matroid.

The classical problem of maximizing a single monotone submodular function $f:2^{V}\rightarrow \RR_+$ under a matroid constraint $\M=(V,\I)$ is formally stated as: $\max_{A\in \I} f(A)$. Throughout this paper, we say that a polynomial time algorithm $\A$ achieves a $(1-\beta)$ approximation factor for this problem, if it returns a feasible solution $S\in \I$ such that $f(S)\geq (1-\beta)\cdot \max_{A\in \I} f(A)$. When $\A$ is randomized, we say that $\A$ achieves a $(1-\beta)$-approximation in expectation if $\EE[f(S)]\geq (1-\beta)\cdot \max_{A\in \I} f(A)$. The well-known standard greedy algorithm \citep{nemhauser1978analysis} is an example of $\A$ with $\beta = 1/2$ \citep{fisher1978analysis}. We will denote by $\text{time}(\A)$ the running time of $\A$.

In this work, we consider the robust variation of the above submodular maximization problem. That is,  for a matroid $\M=(V,\I)$, and a given collection of $k$ monotone submodular functions $f_i:2^{V}\rightarrow \RR_+$ for  $i \in [k]$,  our goal is to efficiently select a set $S$ that maximizes $\min_{i\in[k]} f_i(S)$. We define a $(1-\epsilon)$-approximately optimal solution $S$ as
 \begin{equation}\label{eq:offline_def}
\min_{i\in [k]} f_i(S) \geq (1-\epsilon)\cdot\max_{A \in \I} \min_{i\in[k]} f_i(A).
 \end{equation}

We also consider  the online variation of the above optimization problem in presence of an adversary. In this setting, we are given a fixed matroid $\M=(V,\I)$. At each time step  $t\in[T]$, we choose a set $S^t$. An adversary then selects a collection of $k$ monotone submodular functions $\{f_i^t\}_{i\in[k]}:2^V\to[0,1]$.
We receive a reward of $\min_{i\in[k]}\EE[f_i^t(S^t)]$, where the expectation is taken over any randomness in choosing  $S^t$. We can then use the knowledge of the adversary's actions, i.e., oracle access to $\{f_i^t\}_{i\in[k]}$, in our future decisions. We consider non-adaptive adversaries whose choices $\{f_i^t\}_{i\in[k]}$ are independent of $S^\tau$ for $\tau < t$. In other words, an adversarial sequence of functions $\{f_i^1\}_{i\in[k]}, \dots, \{f_i^T\}_{i\in[k]}$ is chosen upfront without being revealed to the  optimization algorithm. Our goal is to design an algorithm that maximizes the total payoff $\sum_{t\in[T]}\min_{i\in[k]}\EE[f_i^t(S^t)]$. Thus, we would like to obtain a cumulative reward that competes with that of the fixed set $S\in \I$  we should have played had we known all the functions $f_i^t$ in advance, i.e., compete with $\max_{S\in \I}\sum_{t\in[T]}\min_{i\in[k]}f^t_i(S).$ As in the offline optimization problem, we also consider competing with $(1-\epsilon)$ fraction of the above benchmark. In this case, $\text{\bf Regret}_{1-\epsilon}(T)$ denotes how far we are from this goal. That is,
\begin{equation}\label{eq:regret}
\text{\bf Regret}_{1-\epsilon}(T)=(1-\epsilon)\cdot\max_{S\in\I}\sum_{t\in[T]} \min_{i\in[k]}f^t_i(S) - \sum_{t\in[T]} \min_{i\in[k]}\EE\left[f^t_i(S^t)\right].
\end{equation}
We desire algorithms whose $(1-\epsilon)$-regret is sublinear in $T$. That is, we get arbitrarily close to a $(1-\epsilon)$ fraction of the benchmark as $T\rightarrow \infty$.

The offline (Equation~\ref{eq:offline_def}), or online (Equation~\ref{eq:regret}) variations of  robust monotone submodular functions, are known to be NP-hard  to approximate to any polynomial factor when the algorithm's choices are restricted to the family of independent sets $\I$~\citep{krause2008robust}. Therefore, to obtain any reasonable approximation guarantee we need to relax the algorithm's constraint set. Such an approximation approach is called a \emph{bi-criteria} approximation scheme in which the algorithm  outputs a set with a \emph{nearly optimal objective value}, while ensuring that the set used is the \emph{union of only a few independent sets in $\I$.}
More formally, to get a $(1-\epsilon)$-approximate solutions, we may use a set $S$ where $S = S_1\cup \dots\cup S_\ell$ such that $S_1,\dots,S_\ell \in \I$ and $\ell$ is a function of $\frac{1}{\epsilon}$ and other parameters. Since the output set $S$ is possibly infeasible, we define the \emph{violation ratio} $\nu$ as the minimum number of feasible sets whose union is $S$. To exemplify this, consider partition constraints: here we are given a partition \(\{P_1, \ldots, P_q\}\) of the ground set and the goal is to pick a subset that includes at most \(b_j\) elements from part \(P_j\) for each \(j\). Then, the union of $\ell$ feasible sets have at most \(\ell\cdot b_j\) elements in each part. In our example, the violation ratio corresponds to $\nu = \max_{j\in[q]}\lceil|S\cap P_j|/ b_j\rceil$.

\subsection{Our Results and Contributions}\label{sec:results}
We present (nearly tight)  bi-criteria approximation algorithms for the offline and online variations of  robust monotone submodular optimization under matroid constraints.  Throughout the paper, we assume that the matroid is accessible via an independence oracle and the submodular functions are accessible via a value oracle. Moreover, we use $\log_a$ to denote logarithm with base $a$ (when the subscript is not explicit we assume is base 2, i.e, $\log := \log_2$) and $\ln$ to denote the natural logarithm.

For the offline setting of the problem we obtain the following general result:
\begin{theorem}\label{theorem1:offline}
Consider a polynomial time $(1-\beta)$-approximation algorithm $\A$ for the problem of maximizing a single monotone submodular function subject to a matroid constraint. Then, for the offline robust submodular optimization problem \eqref{eq:offline_def}, \nhedit{for any $0<\epsilon<1$,} there is a polynomial time bicriteria algorithm that uses \emph{\texttt{ext}-$\A$ as a subroutine, runs} in \[O\left(\text{\emph{time}}(\A)\cdot \log_{1/\beta}\left(\frac{k}{\epsilon}\right) \cdot\log(n)\cdot\min\left\{\frac{nk}{\epsilon},\log_{1+\epsilon} (\max_{e,j} f_j(e))\right\}\right)\] time and returns a set $S^{\alg}$, such that
\[\EE\left[\min_{i\in[k]} f_i(S^{\alg})\right]\geq (1-\epsilon) \cdot \max_{S\in \I} \min_{j\in[k]} f_j(S),\]
where the expectation is taken over any randomization of $\A$, $S^{\alg}=S_1\cup \dots\cup S_\ell$ with $\ell = \lceil\log_{1/\beta} \frac{k}{\epsilon}\rceil$, and $S_1,\dots,S_\ell \in \I$.
\end{theorem}
The subroutine \texttt{ext}-$\A$ that achieves this result is an extended version of the algorithm $\A$. Since $\A$ outputs a feasible set, then \texttt{ext}-$\A$ reuses $\A$ in an iterative scheme, so that it generates a \emph{small family} of feasible sets whose union achieves the \((1-\epsilon)\)-guarantee. The argument is reminiscent of a well-known fact for submodular function maximization under cardinality constraints: letting the standard greedy run longer results in better approximations at the expense of violating the cardinality constraint. Our extended algorithm \texttt{ext}-$\A$ works in a similar spirit, however it iteratively produces feasible sets in the matroid. This framework generalizes the idea presented in \citep{torrico_etal19}. We emphasize that our procedure does not correspond to an extension of the algorithm for cardinality constraint presented in \citep{krause2008robust}. The natural extension of their algorithm to a single matroid constraint would be to run an algorithm $\A$ over a \emph{larger} feasibility constraint. However, this approach does not provide any bicriteria approximation. The main challenge for matroid constraints is to define an appropriate notion of violation. In our work, we measure violation by the number of independent sets that are needed to cover the given set, as a contrast of \citet{krause2008robust} who define it in terms of cardinality. Our subroutine \texttt{ext}-$\A$ constructs a family of independent sets and that is pivotal to obtain provable guarantees.

A natural candidate for $\A$ is the standard greedy algorithm, from now on referred as \texttt{Greedy}. Since \texttt{Greedy} achieves a $1/2$ approximation factor \citep{fisher1978analysis}, then the number of feasible sets needed is $\ell =\lceil\log_{2} \frac{k}{\epsilon}\rceil$ \citep{torrico_etal19}. Unfortunately, \texttt{Greedy}'s running time $O(n\cdot r)$ depends on the rank of the matroid $r$. This can be computationally inefficient when $r$ is sufficiently large. To improve this, in Section \ref{sec:offline_analysis} we propose an extended version of the \emph{threshold greedy} algorithm introduced by \citet{badanidiyuru_vondrak14}. The main advantage of the threshold greedy algorithm, further referred as \texttt{ThGreedy}, is its running time which does not depend on the rank of the matroid. Formally, we obtain the following corollary.
\begin{corollary}\label{corollary:extended_th_greedy}
For the offline robust submodular optimization problem \eqref{eq:offline_def}, for any $0<\epsilon,\delta<1$, there is a polynomial time bicriteria algorithm that uses \emph{\texttt{ext-ThGreedy}} as a subroutine, runs in \[O\left(\frac{n}{\delta}\cdot\log\left(\frac{n}{\delta}\right) \cdot \log_2\left(\frac{k}{\epsilon}\right) \cdot\log(n)\cdot\min\left\{\frac{nk}{\epsilon},\log_{1+\epsilon} (\max_{e,j} f_j(e))\right\}\right)\] time and returns a set $S^{\alg}$ such that
\[\min_{i\in[k]} f_i(S^{\alg})\geq (1-\epsilon) \cdot \max_{S\in \I} \min_{j\in[k]} f_j(S),\]
where $S^{\alg}=S_1\cup \dots\cup S_\ell$ with $\ell = \lceil\log_2 \frac{k}{\epsilon}\rceil$, and $S_1,\dots,S_\ell \in \I$.
\end{corollary}

To achieve a tight bound on the size of the family $\ell$, we use an improved version of the continuous greedy algorithm \citep{vondrak2008optimal, badanidiyuru_vondrak14} as the inner algorithm $\A$. Since the continuous greedy algorithm achieves a $1-1/e$ approximation ratio, we need $\ell = \lceil\ln \frac{k}{\epsilon}\rceil$ feasible sets to achieve a $1-\epsilon$ fraction of the true optimum, which matches the hardness result presented in \citep{krause2008robust}. This bi-criteria algorithm is much simpler than the one presented in \citep{torrico_etal19}. We present the main offline results and the corresponding proofs in Section~\ref{sec:offline}.

One might hope that similar results can be obtained even when functions are non-monotone (but still submodular).  As we show in Section \ref{sec:necc_mon} this is not possible. To support our theoretical guarantees we provide an exhaustive computational study in Section \ref{sec:offline_experiments}. In this part, we observe that the main computational bottleneck of the bi-criteria algorithms is to \emph{certify} near-optimality of the output solution. To solve this, we present significant implementation improvements such as lazy evaluations and an early stopping criterion, which empirically show how the computational cost can be drastically improved.

Our offline approach is quite flexible in the sense that Theorem \ref{theorem1:offline} uses an arbitrary algorithm $\A$. This allows us to consider different algorithms in the literature of submodular optimization and further extend our results to other classes of constraints, such as knapsack constraints or multiple matroids. We describe these extensions in Section~\ref{sec:extensions}.
 
A natural question is whether our algorithm can be carried over into the online setting, where functions are revealed over time. 
For the online variant, we present the following result:
\begin{theorem}\label{theorem:online}
For the online robust submodular optimization problem \eqref{eq:regret},
\nhedit{for any $0<\epsilon<1$,} there is a randomized polynomial time algorithm that returns a set $S^t$ for each stage $t\in [T]$, we get
\begin{equation*}\sum_{t\in[T]}    \min_{i\in [k]} \EE\left[f^t_i(S^t)\right] \geq   (1-\epsilon) \cdot \max_{S\in \I}  \sum_{t\in[T]} \min_{i\in [k]} f^t_i(S)  -O\left((1-\epsilon)n^{\frac{5}{4}}\sqrt{T}\right),\end{equation*}
 where $S^t=S^t_1\cup \dots\cup S^t_\ell$ with $\ell = \lceil\ln\frac{1}{\epsilon}\rceil$, and $S^t_1,\dots,S^t_\ell \in \I$.
\end{theorem}
We remark that the guarantee of Theorem~\ref{theorem:online} holds with respect to the minimum of $\EE[f_i^t(S^t)]$, as opposed to the guarantee of Theorem~\ref{theorem1:offline} that directly bounds the minimum of $f_i(S)$. Therefore, the solution for the online algorithm is a union of only $\lceil\ln \frac{1}{\epsilon}\rceil$ independent sets, in contrast to the offline solution which is the union of $\lceil\log \frac{k}{\epsilon}\rceil$ independent sets. The main challenge in the online algorithm is to deal with non-convexity and non-smoothness due to submodularity exacerbated by the robustness criteria.   Our approach to coping with the robustness criteria is to use the \emph{soft-min}  function $-\frac{1}{\alpha}\ln \sum_{i\in[k]} e^{-\alpha g_i}$, defined for a collection of smooth functions $\{g_i\}_{i\in[k]}$ and a suitable parameter $\alpha>0$. This function is also known as \emph{log-sum-exp}; for some of its properties and applications we refer the interested reader to \citep{calafiore_ghaoui14}. While the choice of the specific soft-min function is seemingly arbitrary, one feature is crucial for us: its gradient is a convex combination of the gradients of the $g_i$'s. Using this observation, we use parallel instances of the Follow-the-Perturbed-Leader (FPL) algorithm, presented by \citet{kalai_etal05}, one for each discretization step in the continuous greedy algorithm. We believe that the algorithm might be of independent interest to perform online learning over a minimum of several functions, a common feature in robust optimization. The main result and a its proof appears in Section~\ref{sec:online_problem}.

\subsection{Related Work}

Building on the classical work of \citet{nemhauser1978analysis}, constrained submodular maximization problems have seen much progress recently (see for example \citep{calinescu2011maximizing,chekuri2010dependent,buchbinder2014submodular,buchbinder2016comparing}). Robust submodular maximization generalizes submodular function maximization under a matroid constraint for which a $(1-\frac1e)$-approximation is known~\citep{calinescu2011maximizing} and is optimal. The problem has been studied for constant $k$ by \citet{chekuri2010dependent} who give a $(1-\frac1e-\epsilon)$-approximation algorithm with running time $O\left(n^{\frac{k}{\epsilon}}\right)$. Closely related to our problem is the submodular cover problem where we are given a submodular function $f$, a target $b\in \RR_+$, and the goal is to find a set $S$ of minimum cardinality such that $f(S)\geq b$. A simple reduction shows that robust submodular maximization under a cardinality constraint reduces to the submodular cover problem~\citep{krause2008robust}. \citet{wolsey1982analysis} showed that the greedy algorithm gives an $O(\ln{\frac{n}{\epsilon}})$-approximation, where the output set $S$ satisfies $f(S)\geq (1-\epsilon)b$. \citet{krause2008robust} use this approximation to build a bi-criteria algorithm for the cardinality case which achieves tight bounds. This approach falls short of achieving a bi-criteria approximation when the problem is defined over a matroid. A natural extension of their approach to matroid constraints would be to run a single algorithm $\A$ over a \emph{larger} feasibility constraint. However, this procedure does not provide any guarantee. Therefore, the main challenge is to define an appropriate notion of violation. In this work, we measure violation not by cardinality, as in \citep{krause2008robust}, but by the number of independent sets that are needed to cover the given set. Our subroutine \texttt{ext}-$\A$ picks one feasible set at a time and that is crucial to obtain provable guarantees. \citet{powers2016constrained} considers the same robust problem with matroid constraints. However, they take a different approach by presenting a bi-criteria algorithm that outputs a feasible set that is good only for a fraction of the $k$ monotone submodular functions. A deletion-robust submodular optimization model is presented in \citep{krause2008robust}, which is later studied by \citet{orlin_etal16,bogunovic_etal17,kazemi_etal18}. Influence maximization~\citep{kempe_etal03} in a network has been a successful application of submodular maximization and recently, \citet{he_etal16} and \citet{chen_etal16} study the robust influence maximization problem. Robust optimization for non-convex objectives (including submodular functions) has been also considered by \citet{chen2017robust}, however with weaker guarantees than ours due to the extended generality. Specifically, their algorithm outputs $\frac{r\log k}{\epsilon^2\OPT}$ feasible sets whose union achieves a factor of $(1-1/e - \epsilon)$. Finally, \citet{wilder17} studies a similar problem in which the set of feasible solutions is the set of all distributions over independent sets of a matroid. In particular, for our setting \citet{wilder17} gives an algorithm that outputs $O(\frac{\log k}{\epsilon^3})$ feasible sets whose union obtains $(1-1/e)^2$ fraction of the optimal solution. Our results are stronger than the ones obtained by \citet{chen2017robust} and \citet{wilder17}, since we provide the same guarantees using the union of fewer feasible sets. Other variants of the robust submodular maximization problem are studied by \citet{mitrovic_etal18,staib_etal18}.

There has been some prior work on online submodular function maximization that we briefly review here. \citet{streeter_etal08} study the {\it budgeted maximum submodular coverage} problem and consider several feedback cases (denote $B$ a integral bound for the budget): in the full information case, a $(1-\frac{1}{e})$-expected regret of $O(\sqrt{BT\ln n})$ is achieved, but the algorithm uses $B$ experts which may be very large.
In a follow-up work, \citet{golovin_etal14} study the online submodular maximization problem under partition constraints, and then they generalize it to general matroid constraints.
For the latter one, the authors present an online version of the continuous greedy algorithm, which relies on the Follow-the-Perturbed-Leader algorithm of \citet{kalai_etal05} and obtain a $(1-\frac{1}{e})$-expected regret of $O(\sqrt{T})$. Similar to this approach, our bi-criteria online algorithm will also use the Follow-the-Perturbed-Leader algorithm as a subroutine. Recent results on other variants of the online submodular maximization problem are studied in \citep{soma19a} and \citep{zhang_etal19}.

%% file: greedy.tex
\section{The Offline Case}\label{sec:offline}
In this section, we consider offline robust optimization (Equation~\ref{eq:offline_def})  under matroid constraints.

\subsection{General Offline Algorithm and Analysis}\label{sec:offline_analysis}
In this section, we present a general procedure to achieve a (nearly) tight bi-criteria approximation for the problem of interest and prove \Cref{theorem1:offline}. 

First, consider a non-negative monotone submodular function $g:2^V\to\RR_+$, a matroid $\M=(V,\I)$, and a polynomial time $(1-\beta)$-approximation algorithm $\A$ for the  problem of maximizing $g$ over $\M$. Formally, in the deterministic case, $\A$ outputs a feasible set $S \in \I$ such that
\(g(S)\geq (1-\beta)\cdot\max_{A\in \I}g(A).\)
If $g(\emptyset)\neq 0$, then we define a new function $g':2^V\to\RR_+$ as $g'(A) := g(A) - g(\emptyset)$, which remains being monotone and submodular. The approximation guarantee in this case is
\(g(S)-g(\emptyset)\geq (1-\beta)\cdot\max_{A\in \I}\{g(A)-g(\emptyset)\}.\)
When $\A$ is a randomized algorithm, we say that $\A$ achieves a $(1-\beta)$ factor in expectation, if $\A$ outputs a random feasible set $S\in \I$ such that
\(\EE[g(S)-g(\emptyset)]\geq (1-\beta)\cdot\max_{A\in \I}\{g(A)-g(\emptyset)\}.\)
We define  in Algorithm \ref{alg:general_ext_alg} our main procedure \texttt{ext}-$\A$ as an \emph{extended version} of $\A$ that runs iteratively $\ell\geq 1$ times.
\begin{algorithm}
\caption{General Extended Algorithm for Submodular Optimization, \texttt{ext}-$\A$}\label{alg:general_ext_alg}
\begin{algorithmic}[1]
\Require  $\ell\geq 1$, a monotone submodular function $g:2^{V} \rightarrow \RR_+$, a matroid $\M=(V,\I)$ and algorithm $\A$.
\Ensure sets $S_1,\ldots, S_\ell\in \I$.
\For {$\tau=1, \dots, \ell$}
\State Define $\tilde{g}(S) = g(S\bigcup\cup_{j=1}^{\tau-1} S_j )$.
\State $S_\tau\gets \A(\tilde{g},\M)$.
\EndFor
\end{algorithmic}
\end{algorithm}

Note that function $\tilde{g}$ defined in line 2 of Algorithm \ref{alg:general_ext_alg} is also monotone and submodular. Observe that we can recover algorithm $\A$ by simply considering $\ell = 1$ in \texttt{ext}-$\A$. More importantly, we obtain the following guarantee for \texttt{ext}-$\A$.
\begin{theorem}
\label{thm:extended-gen-alg}
Consider a monotone submodular function $g:2^V\to\RR_+$ with $g(\emptyset)=0$, a matroid $\M=(V,\I)$, and a polynomial time $(1-\beta)$-approximation algorithm $\A$ for the problem of maximizing $g$ over $\M$. For any $\ell\geq 1$, Algorithm \ref{alg:general_ext_alg} returns sets $S_1,\ldots, S_\ell$ such that
\[\EE\left[g\left(\cup_{\tau=1}^\ell S_\tau\right)\right]\geq \left(1-\beta^\ell\right) \cdot \max_{S\in \I} g(S),\]
where the expectation is taken over any randomization of $\A$ when choosing $S_1,\ldots, S_\ell$.
\end{theorem}
\begin{proof} Let us assume that $\A$ is deterministic; in the randomized case the proof follows similarly by taking the corresponding expectations. From the first iteration of Algorithm \ref{alg:general_ext_alg} and using the guarantees of $\A$ we conclude that $g(S_1)-g(\emptyset)\geq \left(1-\beta\right)\cdot\max_{S\in \I}  \left\{g(S)-g(\emptyset)\right\}$.
We use the above statement to prove our theorem by induction. For $\tau=1$, the claim follows directly. Consider any $\ell\geq 2$. Observe that the algorithm in iteration $\tau=\ell$, is exactly algorithm $\A$ run on submodular function $\tilde{g}:2^{V} \rightarrow \RR_+$ where $\tilde{g}(S):= g(S\bigcup\cup_{j=1}^{\ell-1} S_j )$. This procedure returns $S_\ell$ such that
$\tilde{g}(S_\ell)- \tilde{g}(\emptyset)\geq \left(1-\beta\right)\cdot \max_{S\in \I} \left\{\tilde{g}(S)-\tilde{g}(\emptyset)\right\},$
which implies that
\begin{equation*} g\left(\cup_{\tau=1}^\ell S_\tau\right)-g\left(\cup_{\tau=1}^{\ell-1}S_\tau\right)\geq  \left(1-\beta\right)\cdot\max_{S\in \I}\left\{g(S) -   g\left(\cup_{\tau=1}^{\ell-1}S_\tau\right)\right\}. \end{equation*}
By induction we know $ g\left(\cup_{\tau=1}^{\ell-1}S_\tau\right)\geq \left(1-\beta^{\ell-1}\right)\cdot\max_{S\in \I} g(S).$ Thus, we obtain
\begin{equation*}
g\left(\cup_{\tau=1}^{\ell}S_\tau\right)\geq (1-\beta)\cdot \max_{S\in \I}g(S)+\beta\cdot  g\left(\cup_{\tau=1}^{\ell-1}S_\tau\right)\geq \left(1-\beta^\ell\right)\cdot\max_{S\in \I} g(S).
\end{equation*}
\end{proof}
We now apply Theorem~\ref{thm:extended-gen-alg} for the robust submodular
problem, in which we are given monotone submodular functions $f_i:2^V\to\RR_+$ with $f_i(\emptyset) = 0$ for $i\in[k]$. Our main bicriteria algorithm consists in two consecutive steps: (1) get an estimate $\gamma$ of the value $\OPT$ and (2) apply subroutine \texttt{ext}-$\A$ to a convenient function depending on $\gamma$. Formally, we obtain an estimate $\gamma$ on the value of the optimal solution $\OPT:=  \max_{S\in \I} \min_{i\in[k]} f_i(S)$ via a binary search. For the purpose of the proof of Theorem \ref{theorem1:offline}, given parameter $\epsilon>0$, let us assume that $\gamma$ has relative error $1-\frac{\epsilon}{2}$, i.e., \(\left(1-\frac{\epsilon}{2}\right)\OPT \leq \gamma\leq \OPT.\) As in \citep{krause2008robust}, let $g:2^{V}\rightarrow \RR_+$ be defined for any $S\subseteq V$ as follows
\begin{equation}\label{eq:average_function} g(S):= \frac{1}{k} \sum_{i\in[k]} \min\{f_i(S), \gamma\}.\end{equation}
Observe that $\max_{S\in \I} g(S) = \gamma$ whenever $\gamma\leq\OPT$. Moreover, note that $g$ is also a monotone submodular function. Therefore, the second step of the bicriteria algorithm is to run algorithm \texttt{ext}-$\A$ on the function $g$ to obtain a candidate solution. A more detailed description of the algorithm can be found in Section \ref{sec:offline_experiments}.
\begin{proof}[Proof of Theorem \ref{theorem1:offline}]
  We assume $\A$ to be deterministic; the randomized case can be easily proved by consider the proof below for each realization sequence of $S_1,\ldots, S_\ell$. Consider the family of monotone submodular functions \(\{f_i\}_{i \in [k]}\) and define $g$ as in equation \eqref{eq:average_function} using parameter $\gamma$ with relative error of \(1-\frac{\epsilon}{2}\) . If we run Algorithm \ref{alg:general_ext_alg} on $g$ with $\ell\geq \lceil \log_{1/\beta} \frac{2k}{\epsilon} \rceil$, we get a set $S^{\alg}= S_1 \cup\cdots \cup S_\ell$, where $S_j\in\I$ for all $j\in[\ell]$. Moreover, Theorem~\ref{thm:extended-gen-alg} implies that
\[g(S^{\alg})\geq \left(1-\beta^\ell\right)\cdot\max_{S\in \I}g(S) \geq \left(1-\frac{\epsilon}{2k}\right)\cdot\gamma.\]
Now, we will prove that $f_i(S^{\alg})\geq \left(1-\frac{\epsilon}{2}\right)\cdot \gamma$, for all $i\in [k]$. Assume by contradiction that there exists an index $i^*\in[k]$ such that $f_{i^*}(S^{\alg})< \left(1-\frac{\epsilon}{2}\right)\cdot \gamma$. Since, we know that $\min\{f_i(S^{\alg}),\gamma\}\leq\gamma$ for all $i\in[k]$, then
\begin{equation*}
g(S^{\alg}) \leq \frac{1}{k}\cdot f_{i^*}(S^{\alg}) + \frac{k-1}{k}\cdot\gamma <  \frac{1-\epsilon/2}{k}\cdot\gamma + \frac{k-1}{k}\cdot\gamma = \left(1-\frac{\epsilon}{2k}\right)\cdot\gamma,\end{equation*}
contradicting $ g(S^{\alg}) \geq\left(1-\frac{\epsilon}{2k}\right)\cdot\gamma $. Therefore, we obtain
$f_i(S^{\alg})\geq \left(1-\frac{\epsilon}{2}\right)\cdot\gamma \geq (1-\epsilon)\cdot\OPT$,
for all $i\in[k]$ as claimed.
\end{proof}

\paragraph{Running time analysis.} \label{sec:running_time}
In this section, we study the running time of the bi-criteria algorithm we just presented. To show that a set of polynomial size of values for $\gamma$ exists such that one of them satisfies $(1-\epsilon/2)\OPT\leq \gamma\leq \OPT$, we simply try $ \gamma = nf_i(e)(1-\epsilon/2)^j $ for all $i\in [k]$, $e\in V$, and $j=0,\dots,\lceil \ln_{1-\epsilon/2}(1/n)\rceil$. Note that there exists an index $i^*\in [k]$ and a set $S^*\in \I$ such that $\OPT=f_{i^*}(S^*)$. Now let $e^*=\argmax_{e\in S^*}f_{i^*}(e)$. Because of submodularity and monotonicity we have $\frac{1}{\lvert S^*\rvert}f_{i^*}(S^*)\leq f_{i^*}(e^*)\leq f_{i^*}(S^*)$. So, we can conclude that $1\geq \OPT/nf_{i^*}(e^*)\geq 1/n$, which implies that $j=\lceil \ln_{1-\epsilon/2}(\OPT/nf_{i^*}(e^*))\rceil$ is in the correct interval, obtaining
\[ (1-\epsilon/2)\OPT \leq nf_{i^*}(e^*)(1-\epsilon/2)^j\leq \OPT. \]
We remark that the dependency of the running time on $\epsilon$ can be made logarithmic by running a binary search on $j$ as opposed to trying all $j=0,\dots, \lceil \ln_{1-\epsilon/2}(1/n)\rceil$. 
This would take at most $\frac{n k}{\epsilon}\cdot\log n$ iterations. We could also say that doing a binary search to get a value up to a relative error of $1-\epsilon/2$ of $\OPT$ would take $\log_{1+\epsilon}\OPT$. So, we consider the minimum of those two quantities $\min\{\frac{n k}{\epsilon}\cdot \log n,\log_{1+\epsilon}\OPT\}$.
Given that Algorithm \ref{alg:general_ext_alg} runs in $\ell \cdot \text{time}(\A)$ where $\ell = \lceil\log_{1/\beta} \frac{k}{\epsilon}\rceil$ is the number of rounds, we conclude that the bi-criteria algorithm runs in \(O(\text{time}(\A) \cdot\log_{1/\beta}\frac{k}{\epsilon}\cdot\min\{\frac{nk}{\epsilon}\cdot \log n,\log_{1+\epsilon}\OPT\}) \) time.  

\subsubsection{Two Deterministic Classical Algorithms: Greedy and Local Search.}
The most natural candidate for $\A$ is \texttt{Greedy} \citep{nemhauser1978analysis} which achieves a ratio of $1/2$ for the problem of maximizing a single monotone submodular function subject to a matroid constraint \citep{fisher1978analysis}. In this case, we know that $\text{time(\texttt{Greedy})} = O(n\cdot r)$, where $r$ is the rank of the matroid. Using \texttt{Greedy}, we are able to design its extended version \texttt{ext-Greedy}, formally outlined in Algorithm \ref{alg:ext_greedy}. 
\begin{algorithm}
\caption{Extended Greedy Algorithm for Submodular Optimization, \texttt{ext-Greedy}}\label{alg:ext_greedy}
\begin{algorithmic}[1]
\Require  $\ell\geq 1$, monotone submodular function $g:2^{V} \rightarrow \RR_+$, matroid $\M=(V,\I)$.
\Ensure sets $S_1,\ldots, S_\ell\in \I$.
\For {$\tau=1, \dots, \ell$}
\State $S_\tau\gets \emptyset$
\While {$S_\tau$ is not a basis of $\M$}
	 \State {\bf Compute} \( e^* = \argmax_{S_\tau+ e\in \I} \{g(\cup_{j=1}^{\tau} S_j + e)\}. \)
	 \State {\bf Update} $ S_\tau\gets S_\tau+e^*.$
\EndWhile
\EndFor
\end{algorithmic}
\end{algorithm}

From Theorem \ref{theorem1:offline} we can easily derive the following corollary for \texttt{ext-Greedy}:
\begin{corollary}\label{theorem:extended_greedy}
For the offline robust submodular optimization problem \eqref{eq:offline_def}, for any $0<\epsilon,\delta<1$, there is a polynomial time bicriteria algorithm that uses \emph{\texttt{ext-Greedy} as a subroutine, runs} in \[O\left(n\cdot r \cdot \log_2\left(\frac{k}{\epsilon}\right) \cdot\log(n)\cdot\min\left\{\frac{nk}{\epsilon},\log_{1+\epsilon} (\max_{e,j} f_j(e))\right\}\right)\] time and returns a set $S^{\alg}$ such that
\[\min_{i\in[k]} f_i(S^{\alg})\geq (1-\epsilon) \cdot \max_{S\in \I} \min_{j\in[k]} f_j(S),\]
where $S^{\alg}=S_1\cup \dots\cup S_\ell$ with $\ell = \lceil\log_2 \frac{k}{\epsilon}\rceil$, and $S_1,\dots,S_\ell \in \I$.
\end{corollary}
Another natural candidate is the standard local search algorithm, from now on referred as \texttt{LS} \citep{fisher1978analysis}. Roughly speaking, this algorithm starts with a maximal feasible set and iteratively swap elements if the objective is improved while maintaining feasibility. If the objective cannot be longer improved, then the algorithm stops. \citet{fisher1978analysis} prove that this procedure also achieves a $1/2$-approximation, i.e., $\beta = 1/2$. However, the running time of \texttt{LS} cannot be explicitly obtained. Finally, for the offline robust problem the extended version of the local search algorithm, \texttt{ext-LS}, achieves the same guarantees than \texttt{ext-Greedy}, but without an explicit running time.

\subsubsection{Improving Running Time: Extended Threshold Greedy.}\label{sec:th_greedy}

As we mentioned earlier, we are interested in designing efficient bi-criteria algorithms for the robust submodular problem \eqref{eq:offline_def}. Unfortunately, the subroutine \texttt{ext-Greedy} performs $O(n\cdot r\cdot\log_2 \frac{k}{\epsilon})$ function calls, which can be considerably inefficient when $r$ is sufficiently large. Our objective in this section is to study a variant of the standard greedy algorithm that perform less function calls and whose running time does not depend on the rank of the matroid. 
For our purposes, we consider the \emph{threshold greedy} algorithm, from now on referred as \texttt{ThGreedy}, introduced by \citet{badanidiyuru_vondrak14}. Roughly speaking, this procedure iteratively adds elements whose marginal value is above certain threshold while maintaining feasibility. Unlike \texttt{Greedy}, \texttt{ThGreedy} may add more than one element in a single iteration. For cardinality constraint, \citet{badanidiyuru_vondrak14} show that the threshold greedy algorithm achieves a $(1-1/e-\delta)$-approximation factor, where $\delta$ is the parameter that controls the threshold. Moreover, $\text{time}(\text{\texttt{ThGreedy}}) = O(\frac{n}{\delta}\log\frac{n}{\delta})$, which does not depend on the rank of the matroid. We formalize its extended version, \texttt{ext-ThGreedy}, in Algorithm \ref{alg:ext_th_greedy}. The original version corresponds to considering $\ell = 1$.
\begin{algorithm}[h]
\caption{Extended Threshold-Greedy, \texttt{ext-ThGreedy}}\label{alg:ext_th_greedy}
\begin{algorithmic}[1]
\renewcommand{\algorithmicrequire}{\textbf{Input:}}
\renewcommand{\algorithmicensure}{\textbf{Output:}}
\Require $\ell\geq 1$, ground set $V$ with $n:=|V|$, monotone submodular function $g:2^{V} \rightarrow \RR_+$, matroid $\M=(V,\I)$ and $\delta>0$.
\Ensure feasible sets $S_1,\ldots, S_\ell\in \I$.
\For {$\tau=1, \dots, \ell$}
\State $S_\tau\leftarrow\emptyset$
\State $d\leftarrow\max_{e\in V} g(\cup_{j=1}^{\tau-1} S_j + e)$
\For {$(\omega=d; \omega\geq \frac{\delta}{n}d; \omega\leftarrow (1-\delta)\omega)$}
\For {$e\in V\backslash S_\tau$}
	\If {$S_\tau+e \in \I$ and $g_{\cup_{j=1}^{\tau} S_j}(e)\geq \omega$}
	\State $S_\tau\leftarrow S_\tau + e$
	\EndIf
\EndFor
\EndFor
\EndFor
\end{algorithmic}
\end{algorithm}
Similarly than \citep{badanidiyuru_vondrak14}, we can prove the following guarantee of \texttt{ThGreedy} when is used for the problem of maximizing a single monotone submodular function subject to a matroid constraint.
\begin{corollary}\label{cor:threshold_greedy}
For any $\delta>0$, \emph{\texttt{ThGreedy}} achieves a $\left(1- \frac{1}{2-\delta} \right) $-approximation for the problem of maximizing a single monotone submodular function subject to a matroid constraint, using $O(\frac{n}{\delta}\cdot \log\frac{n}{\delta})$ queries.
\end{corollary}
For a detailed proof of Corollary \ref{cor:threshold_greedy}, we refer the interested reader to the Appendix. Given the previous result, we easily obtain Corollary \ref{corollary:extended_th_greedy} using $\beta = \frac{1}{2-\delta}$ in Theorem \ref{theorem1:offline}. The most relevant feature of \texttt{ext-ThGreedy} is its running time which is independent on the rank of the matroid.

\subsubsection{Tight Bounds: Extended Continuous Greedy.}\label{sec:cont_greedy}
To achieve a tight bound on the number of feasible sets in Theorem \ref{theorem1:offline}, we need to make use of the \emph{continuous greedy algorithm} \citep{vondrak2008optimal}, from now on referred as \texttt{CGreedy}. Before explaining the algorithm, let us recall some preliminary definitions. We denote  the indicator vector of a set $S\subseteq V$ by $\one_S\in\{0,1\}^V$, where $\one_S(e) = 1$ if $e\in S$ and zero otherwise; and the matroid polytope by $\P(\M)  = \mathrm{conv}\{\one_S \mid S \in \I\}$. For any non-negative set function $g:2^V\to\RR_+$, its {\it multilinear extension} $G: [0,1]^V \rightarrow \RR_+$ is defined for any $y\in [0,1]^V$ as the expected value of $g(S_y)$, where $S_y$ is the random set generated by drawing independently each element $e\in V$ with probability $y_e$. Formally,
\begin{equation}\label{eq:ML_def} G(y)  = \EE_{S\sim y}[g(S)]=\sum_{S \subseteq V} g(S) \prod_{e \in S} y_e \prod_{e \notin S} (1-y_{e}).\end{equation}
Observe, this is in fact an extension of $g$, since for any subset $S\subseteq V$, we have $g(S)=G(\one_S)$. For any $x,y\in[0,1]^V$, we will denote $x\vee y$ the vector whose components are $[x\vee y]_e = \max\{x_e,y_e\}$.
\begin{fact}{\emph{\citep{calinescu2011maximizing}}.}\label{fact:ML} Let $g$ be a monotone submodular function and $G$ its multilinear extension:
  \begin{enumerate}
  \item  By monotonicity of $g$, we have $\frac{\partial G}{\partial y_e} \geq 0$ for any $e\in V$. This implies that for any $x\leq y$ coordinate-wise, $G(x)\leq G(y)$. On the other hand, by submodularity of $g$, $G$ is concave in any positive direction, i.e., for any $e_1, e_2\in V$ we have $\frac{\partial^2 G}{\partial y_{e_1} \partial y_{e_2}} \leq 0$.

\item Throughout the rest of this paper we will denote by $\nabla_e G(y):=\frac{\partial G}{\partial y_e}$, and 
\begin{equation}\label{eq:delta_ML}\Delta_eG(y):=\EE_{S\sim y}[g_S(e)].\end{equation} It is easy to see that $\Delta_eG(y)=(1-y_e)\nabla_e G(y)$. Moreover, for any $x,y\in[0,1]^V$ it is easy to prove by using submodularity that
\begin{equation}\label{eq:grad} 
G(x\vee y)\leq G(x)+ \Delta G(x)\cdot y\leq  G(x) +  \nabla G(x)\cdot y.
\end{equation}
\end{enumerate}
\end{fact}

Broadly speaking, \texttt{CGreedy} works as follows: the algorithm starts with the empty set $y_0 = 0$ and for every $t\in[0,1]$ continuously finds a feasible direction $z$ that maximizes $\nabla G(y_t)\cdot z$ over $\P(\M)$, where $y_t$ is the current fractional point. Then, \texttt{CGreedy} updates $y_t$ according to $z$. Finally, the algorithm outputs a feasible set by rounding $y_1$ according to pipage rounding \citep{ageev_sviridenko04}, randomized pipage rounding or randomized swap rounding \citep{chekuri2010dependent}. All these rounding procedures satisfy the following property: if $S$ is the result of rounding $y_1$, then $S\in \I$ and $\EE[g(S)]\geq G(y_1)$, where the expectation is taken over any randomization. 

Notably, \citet{vondrak2008optimal} proved that \texttt{CGreedy} finds a feasible set $S$ such that $\EE[g(S)]\geq (1-1/e)\cdot\max_{A\in\I} g(A)$. Unfortunately, $\text{time}(\texttt{CGreedy}) = O(n^8)$ due to the large number of samples required to accurately evaluate the multilinear extension. This running time can be substantially improved by using an accelerated version of the continuous greedy (\texttt{ACGreedy}) introduced by \citet{badanidiyuru_vondrak14}. \texttt{ACGreedy} generalizes the idea of \texttt{ThGreedy} to the continuous framework and outputs a random feasible set $S$ such that $\EE[g(S)]\geq (1-1/e-\delta)\cdot\max_{A\in\I} g(A)$, where $\delta$ is the parameter that controls the threshold. More importantly,  \texttt{ACGreedy} runs in $O(rn\delta^{-4}\log^2 n)$, where $r$ is the rank of the matroid. Therefore, by using \texttt{ext-ACGreedy} with $\beta = 1/e+\delta$ in Theorem \ref{theorem1:offline} we obtain the following corollary
\begin{corollary}\label{cor:extended_acel_cont_greedy}
For the offline robust submodular optimization problem \eqref{eq:offline_def}, for any $0<\epsilon<1$,  there is a polynomial time bicriteria algorithm that uses \emph{\texttt{ext-ACGreedy}} as a subroutine, runs in \[O\left(r\cdot n\cdot\delta^{-4}\cdot\log^2(n)\cdot \ln\left(\frac{k}{\epsilon}\right) \cdot\log(n)\cdot\min\left\{\frac{nk}{\epsilon},\log_{1+\epsilon} (\max_{e,j} f_j(e))\right\}\right)\] time and returns a set $S^{\alg}$ such that
\[\EE\left[\min_{i\in[k]} f_i(S^{\alg})\right]\geq (1-\epsilon) \cdot \max_{S\in \I} \min_{j\in[k]} f_j(S),\]
where $S^{\alg}=S_1\cup \dots\cup S_\ell$ with $\ell = \lceil\ln \frac{k}{\epsilon}\rceil$, and $S_1,\dots,S_\ell \in \I$.
\end{corollary}
As we can see the number of independent sets required for obtaining this result $\ell = \lceil\ln \frac{k}{\epsilon}\rceil$ is smaller up to a constant than the number of sets obtained by \texttt{ext-Greedy}, $\ell = \lceil\log_2 \frac{k}{\epsilon}\rceil$. More importantly, Corollary \ref{cor:extended_acel_cont_greedy} matches the hardness results given by \citet{krause2008robust}.

%% file: monotonicity.tex
\subsection{Necessity of monotonicity}
\label{sec:necc_mon}

 In light of the approximation algorithms for non-monotone submodular function maximization under matroid constraints (see, for example, \citep{lee2009non}), one might hope that an analogous bi-criteria approximation algorithm could exist for robust non-monotone submodular function maximization. However, we show that even without any matroid constraints, getting any approximation in the non-monotone case is $\text{NP}$-hard.

\begin{lemma}
	Unless $P=NP$, no polynomial time algorithm can output a set $\tilde{S}\subseteq V$ given general submodular functions $f_1,\dots,f_k$ such that $\min_{i\in[k]}f_i(\tilde{S})$ is within a positive factor of $\max_{S\subseteq V}\min_{i\in[k]}f_i(S)$.
\end{lemma}
\begin{proof}
	We use a reduction from \textsc{Sat}. Suppose that we have a \textsc{Sat} instance with variables $x_1,\dots,x_n$. Consider $V=\{1,\dots,n\}$. For every clause in the \textsc{Sat} instance we introduce a nonnegative linear (and therefore submodular) function. For a clause $\bigvee_{i\in A}x_i\vee \bigvee_{i\in B} \overline{x_i}$ define
	\[ f(S) := \lvert S\cap A\rvert+\lvert B\setminus S\rvert. \]
	It is easy to see that $f$ is linear and nonnegative. If we let $S$ be the set of true variables in a truth assignment, then it is easy to see that $f(S)>0$ if and only if the corresponding clause is satisfied. Consequently, finding a set $S$ such that all functions $f$ corresponding to different clauses are positive is as hard as finding a satisfying assignment for the \textsc{Sat} instance.
\end{proof}

%% file: offline_exp.tex
\subsection{Computational Study}\label{sec:offline_experiments}

Our objective in this section is to empirically demonstrate that natural candidate subroutines for the offline robust problem \eqref{eq:offline_def}, such as \texttt{ext-Greedy}, can be substantially improved. For this, we will make use of \texttt{ext-ThGreedy} and other implementation improvements such as lazy evaluations and an early stopping rule. Ultimately, with these tools we are able to design an efficient bi-criteria algorithm.

First, let us recall how the general bi-criteria algorithm works. In an outer loop we obtain an estimate $\gamma$ on the value of the optimal solution $\OPT:=  \max_{S\in \I} \min_{i\in[k]} f_i(S)$ via a binary search. Next, for each guess $\gamma$ we define a new submodular set function as \(g(S):= \frac{1}{k} \sum_{i\in[k]} \min\{f_i(S), \gamma\}.\) Finally, we run Algorithm \ref{alg:general_ext_alg} to obtain a candidate solution. Depending on this result, we update the binary search on $\gamma$, and we iterate. We stop the binary search whenever we get a relative error of $1-\epsilon/2$, namely, $(1-\epsilon/2)\OPT\leq\gamma\leq\OPT$. 

\paragraph{Reducing the number of function calls.}
In Section \ref{sec:th_greedy}, we theoretically showed that by using the subroutine \texttt{ext-ThGreedy} we can get (nearly) optimal objective value using less function evaluations than by using \texttt{ext-Greedy} at cost of producing a slightly bigger family of feasible sets. Therefore, we will use \texttt{ext-ThGreedy} for our efficient bi-criteria algorithm.

\paragraph{Certifying (near)-optimality.} 
The main bottleneck of any bi-criteria algorithm remains obtaining a certificate of (near)-optimality or equivalently, a good upper bound on the optimum. We obtain that the optimum value is at most $\gamma$ whenever running \texttt{ext-$\A$} on function $g$ fails to return a solution of desired objective. Due to the desired accuracy in binary search and the number of steps in the extended algorithm, obtaining good upper bounds on the optimum is computationally prohibitive. We resolve this issue by implementing an early stopping rule in the bi-criteria algorithm.  When running \texttt{ext-$\A$} on  function \(g\) (as explained above) we use the stronger guarantee given in Proposition~\ref{thm:extended-gen-alg}: for any iteration $\tau\in[\ell]$ we obtain a set $S_\tau$ such that $g(\cup_{j=1}^\tau S_j)\geq (1-\beta^\tau)\cdot\gamma$. If in some iteration $\tau\in[\ell]$ the algorithm does not satisfy this guarantee, then it means that $\gamma$ is much larger than $\OPT$. In such case, we stop \texttt{ext-$\A$} and update the upper bound on $\OPT$ to be $\gamma$. This allows us to stop the iteration much earlier since in many real instances $\tau$ is typically much smaller than $\ell$ when $\gamma$ is large. This leads to a drastic improvement in the number of function calls as well as CPU time. 

\paragraph{Lazy evaluations.} All greedy-like algorithms and baselines implemented in this section make use of lazy evaluations \citep{minoux_78}. This means, we keep a list of an upper bound $\rho(e)$ on the marginal gain for each element (initially $\infty$) in decreasing order, and at each iteration, it evaluates the element at the top of the list $e'$. If the marginal gain of this element satisfies $g_S(e')\geq \rho(e)$ for all $e\neq e'$, then submodularity ensures $g_S(e')\geq g_S(e)$. In this way, for example, \texttt{Greedy} does not have to evaluate all marginal values to select the best element.

\paragraph{Bounds initialization.} To compute the initial $\lb$ and $\ub$ for the binary search, we run the lazy greedy \citep{minoux_78} for each function in a small sub-collection $\{f_i\}_{i\in[k']}$, where $k'\ll k$, leading to $k'$ solutions $A^1,\ldots, A^{k'}$ with guarantees $f_i(A^i)\geq (1/2)\cdot \max_{S\in \I}f_i(S)$. Therefore, we set $\ub = 2\cdot \min_{i\in[k']}f_i(A^i)$  and $\lb = \max_{j\in[k']}\min_{i\in[k]} f_i(A^j)$. This two values correspond to upper and lower bounds for the true optimum $\OPT$.

To facilitate the interpretation of our theoretical results, we will consider partition constraints in all experiments: the ground set $V$ is partitioned in $q$ sets $\{P_1,\ldots,P_q\}$ and the family of feasible sets is $\I = \{S: \ |S\cap P_j| \leq b, \ \forall j\in[q]\}$, same budget $b$ for each part. We test five methods: \texttt{prev-extG} the extended greedy algorithm with no improvements, and the rest with improvements: \texttt{ext-Greedy}, \texttt{ext-ThGreedy}, and \texttt{ext-SGreedy}. The last method is an heuristic that uses the \emph{stochastic greedy} algorithm \citep{mirzasoleiman_etal15} adapted to partition constraints (see Appendix). The vanilla version of this algorithm samples a smaller ground set in each iteration and optimize accordingly. Also, we tested the extended version of local search, but due to its bad performance compared to greedy-like algorithms, we do not report the results. For the final pseudo-code of the main algorithm, we refer the interested reader to the Appendix. 

After running the four algorithms, we save the solution $S^{\alg}$ with the largest violation ratio $\nu = \max_{j\in[q]}\lceil|S\cap P_j|/ b\rceil$, and denote by $\tau_{\max} := \left\lceil \nu\right\rceil$. Observe that $S^{\alg} = S_1\cup\ldots\cup S_{\tau_{\max}}$ where $S_\tau\in\I$ for all $\tau\in[\tau_{\max}]$. We consider two additional baseline algorithms (without binary search): Random Selection (\texttt{RS}) which outputs a set $\tilde{S} = \tilde{S}_1\cup\ldots\cup \tilde{S}_{\tau_{\max}}$ such that for each $\tau\in[\tau_{\max}]$: $\tilde{S}_\tau$ is feasible, constructed by selecting elements uniformly at random, and $|\tilde{S}_\tau\cap P_j| = |S_\tau\cap P_j|$ for each part $j\in[q]$. Secondly,  (\texttt{G-Avg}) we run  $\tau_{\max}$ times the lazy greedy algorithm on the average function $\frac1k\sum_{i\in[k]}f_i$ and considering constraints $\I_\tau = \{S: \ |S\cap P_j| \leq |S_\tau\cap P_j|, \ \forall j\in[q]\}$ for each iteration $\tau\in[\tau_{\max}]$.

In all experiments we consider the following parameters: approximation $1-\epsilon = 0.99$, threshold $\delta = 0.1$, and sampling in \texttt{ext-SGreedy} with $\epsilon' = 0.1$. The composition of each part $P_j$ is always uniformly at random from $V$.

\subsubsection{Non-parametric Learning.}
We follow the setup in \citep{mirzasoleiman_etal15}. Let $X_V$ be a set of random variables corresponding to bio-medical measurements, indexed by a ground set of patients $V$. We assume $X_V$ to be a Gaussian Process (GP), i.e., for every subset $S\subseteq V$, $X_S$ is distributed according to a multivariate normal distribution $\mathcal{N}(\muB_S,\SigmaB_{S,S})$, where $\muB_S = (\mu_{e})_{e\in S}$ and $\SigmaB_{S,S} = [\K_{e,e'}]_{e,e'\in S}$ are the prior mean vector and prior covariance matrix, respectively. The covariance matrix is given in terms of a positive definite kernel $\K$, e.g., a common choice in practice is the squared exponential kernel $\K_{e,e'} =\exp(-\|x_e-x_{e'}\|^2_2/h)$. Most efficient approaches for making predictions in GPs rely on choosing a small subset of data points. For instance, in the Informative Vector Machine (IVM) the goal is to obtain a subset $A$ such that maximizes the information gain, \(f(A) = \frac{1}{2}\log\text{det}(\ident + \sigma^{-2}\SigmaB_{A,A})\), which is known to be monotone and submodular \citep{krause_guestrin05}.
In our experiment, we use the Parkinson Telemonitoring dataset \citep{tsanas_etal10} consisting of $n = 5,875$ patients with early-stage Parkinsons disease and the corresponding bio-medical voice measurements with 22 attributes (dimension of the observations). We normalize the vectors to zero mean and unit norm. With these measurements we computed the covariance matrix $\Sigma$ considering the squared exponential kernel with parameter $h=0.75$. 
 For  our robust criteria, we consider $k=20$ perturbed versions of the information gain defined with $\sigma^2=1$, i.e., Problem \eqref{eq:offline_def} corresponds to \(\max_{A\in \I}\min_{i\in[20]}f(A) + \sum_{e\in A\cap \Lambda_i}\eta_e\), where \(f(A)=\frac{1}{2}\log\text{det}(\ident + \SigmaB_{AA})\), \(\Lambda_i \) is a random set of size 1,000 with different composition for each $i\in[20]$, and $\eta\sim[0,1]^V$ is a uniform error vector.

We made 20 random runs considering $q = 3$ parts and budget $b=5$. We report the results in Figures \ref{fig:experiments} (a)-(d). In Figures \ref{fig:experiments}  (a) and (b), we show the performance profiles for the running time and the number of function calls of all methods, respectively. The y-axis corresponds to the fraction of the instances in which a specific method performs less than a multiplicative factor (x-axis) with respect to the best performance. For example, we observe in (b) that in only 20\% of the instances the method \texttt{prev-extG} uses less than 2.5 times the number of function calls used by the best method. We also note that any of the three algorithms clearly outperform \texttt{prev-extG}, either in terms of running time (a) or function calls (b). With this, we show empirically that our implementation improvements help in the performance of the algorithm. We also note that \texttt{ext-SGreedy} is likely to have the best performance. Box-plots for the function calls in Figure \ref{fig:experiments} (d) confirm this fact, since \texttt{ext-SGreedy} has the lowest median. Each method is presented in the x-axis and the y-axis corresponds to the number of function calls (the orange line is the median and circles are outliers). In this figure, we do not present the results of \texttt{prev-extG} because of the difference in magnitude. Finally, in (c) we present the objective values (y-axis) obtained in a single run with respect to the number of feasible sets needed to cover the given set (x-axis). For example, when the set constructed by each method is 2 times the size of a feasible set, most of the procedures have an objective value around 10.  We observe that the stopping rule is useful since the three tested algorithms output a nearly optimal solution earlier (using around 3 times the size of a feasible set) outperforming \texttt{prev-extG} and the benchmarks (which need around 6 times), and moreover, at much less computational cost as we mentioned before.

\begin{figure}[ht]
{\begin{tabular}{ccc}
\def\arraystretch{0.5}
\includegraphics[width=0.3\textwidth]{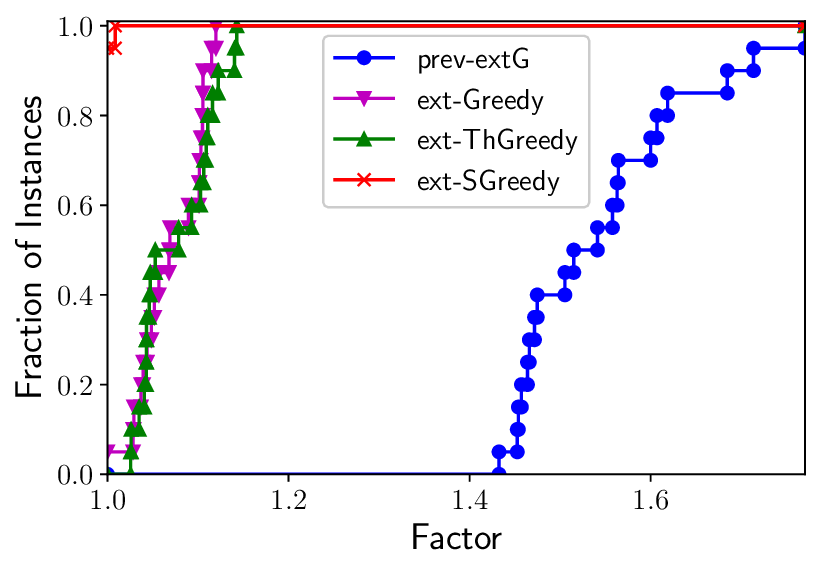} &
\includegraphics[width=0.3\textwidth]{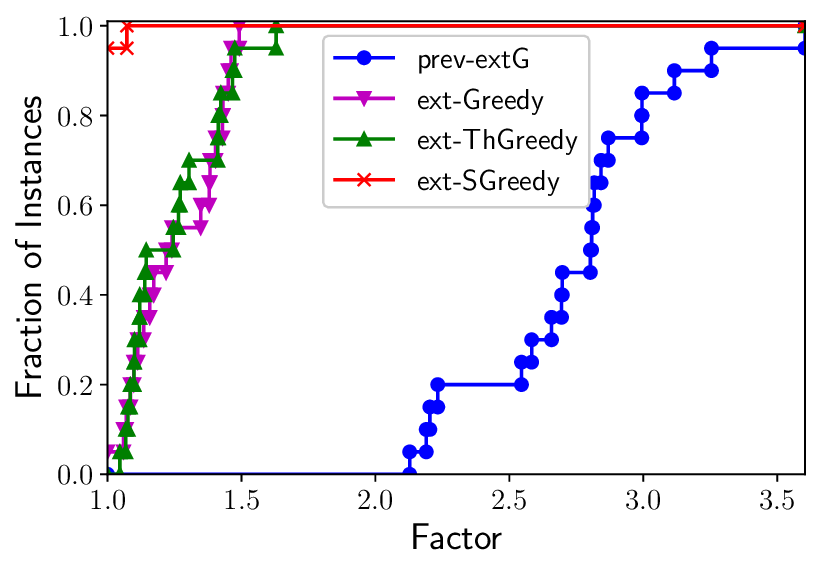} &
\includegraphics[width=0.3\textwidth]{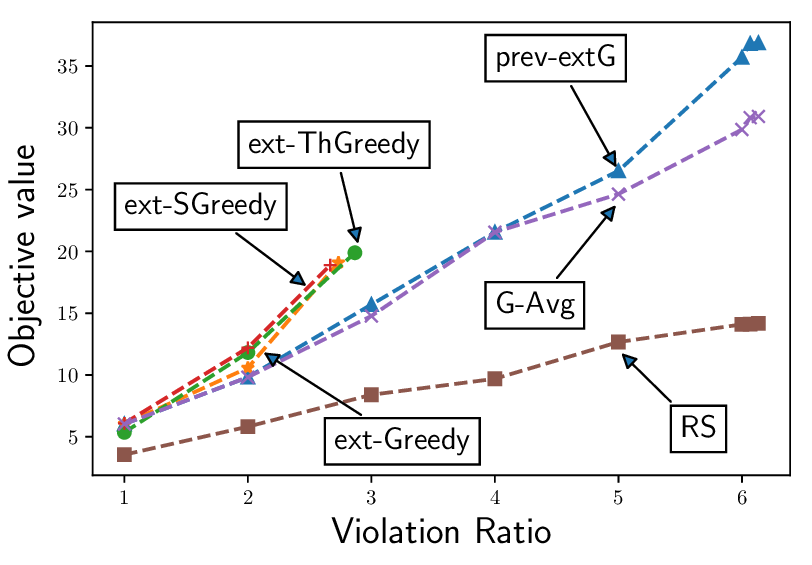} \\
 (a) & (b) & (c) \\
\includegraphics[width=0.3\textwidth]{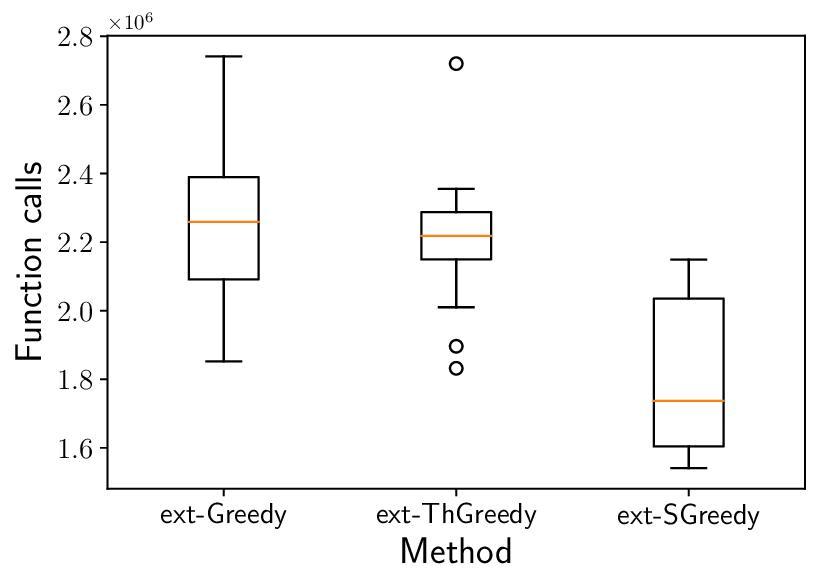} &
\includegraphics[width=0.3\textwidth]{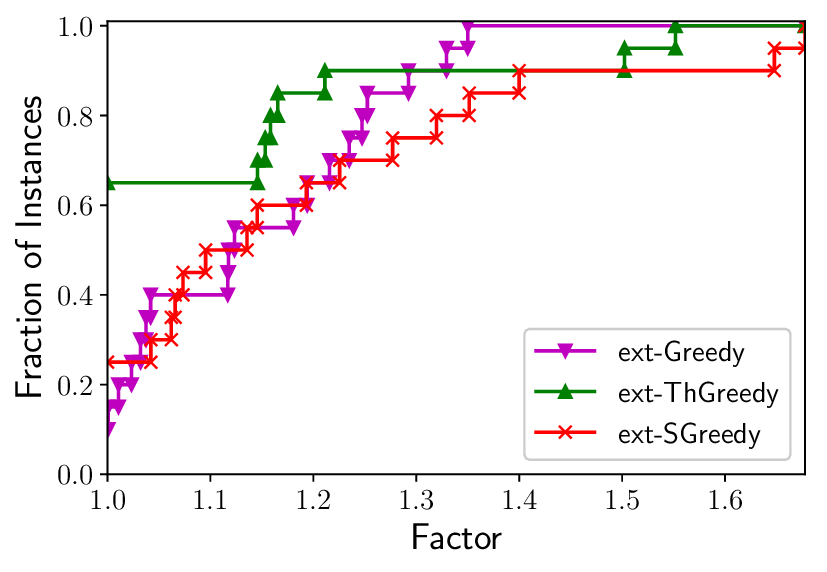}&
\includegraphics[width=0.3\textwidth]{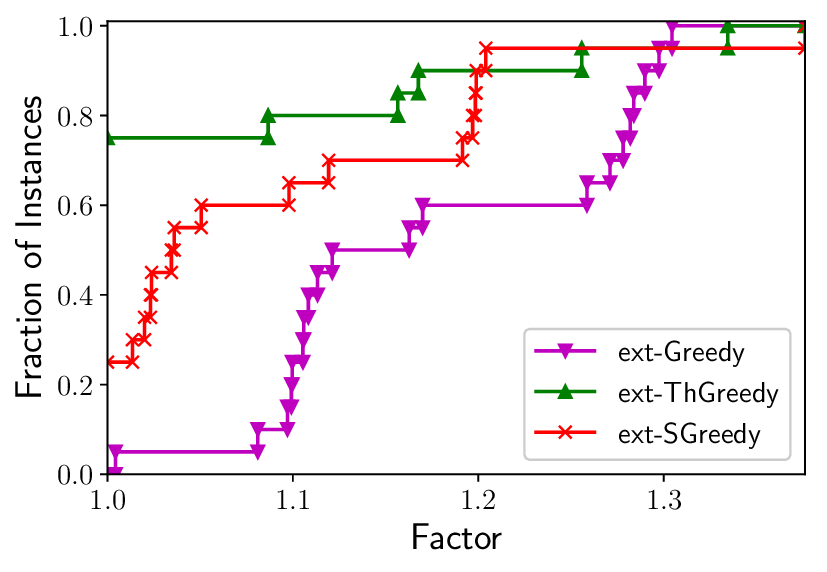}  \\
(d) & (e) & (f) \\
\includegraphics[width=0.3\textwidth]{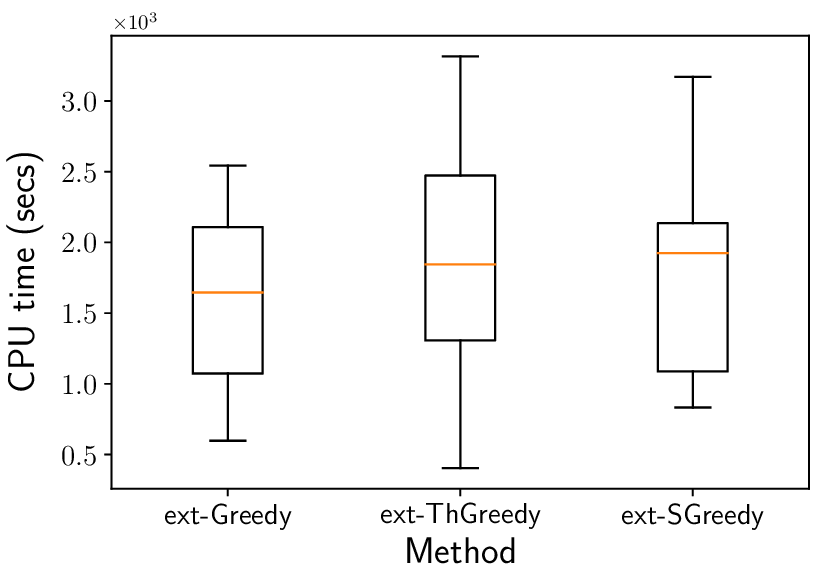} &
\includegraphics[width=0.3\textwidth]{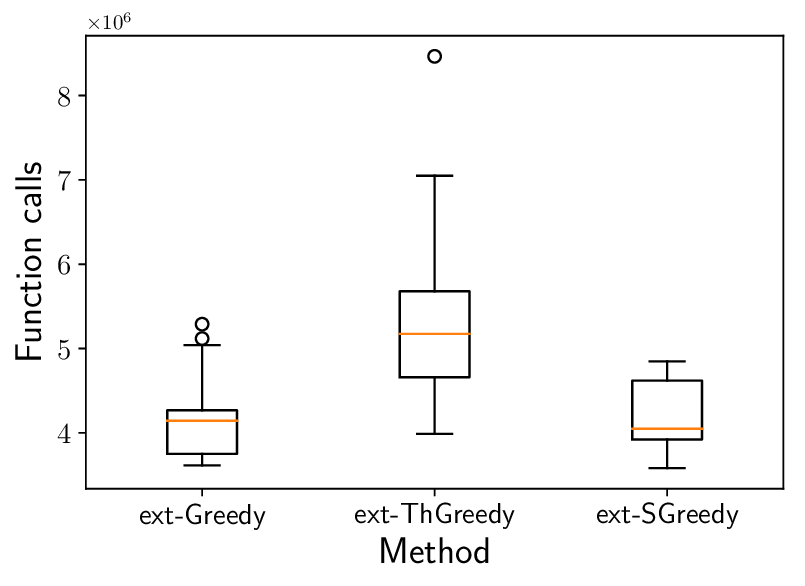} & \\
 (g) & (h) &
 \end{tabular}}
\caption{Experimental Results. Non-parametric learning: performance profiles (a) for running time ($x$-axis is in log-scale; note that \texttt{ext-Greedy} is covered by \texttt{ext-ThGreedy}) and (b) for function calls. In (c) is the objective value versus the violation ratio in a single run of each method. In (d) is the box-plot for the function calls. \emph{Clustering}: (\texttt{small}) performance profiles (e) for the running time and (f) for the function calls. (\texttt{Large}) box-plots (g) for the running time and (h) for the function calls.}
\label{fig:experiments}
\end{figure}

\subsubsection{Exemplar-based Clustering.}
We follow the setup in \citep{mirzasoleiman_etal15}. Solving the \(k\)-medoid problem is a common way to select a subset of exemplars that represent a large dataset $V$ \citep{kaufman_09}. This is done by minimizing the sum of pairwise dissimilarities between elements in $A\subseteq V$ and $V$. Formally, define \(L(A) =\frac{1}{V}\sum_{e\in V} \min_{v\in A}d(e,v)\), where \(d: V\times V\to\RR_+\) is a distance function that represents the dissimilarity between a pair of elements. By introducing an appropriate auxiliary element $e_0$, it is possible to define a new objective \(f(A) := L(\{e_0\}) - L(A+e_0)\) that is monotone and submodular \citep{gomes_krause10}, thus maximizing \(f\) is equivalent to minimizing $L$.
In our experiment, we use the VOC2012 dataset \citep{voc_2012}. The ground set $V$ corresponds to images, and we want to select a subset of the images that best represents the dataset. Each image has several (possible repeated) associated categories such as person, plane, etc. There are around 20 categories in total. Therefore, images are represented by feature vectors obtained by counting the number of elements that belong to each category, for example, if an image has 2 people and one plane, then its feature vector is $(2,1,0,\ldots,0)$ (where zeros correspond to other elements). We choose the Euclidean distance $d(e,e') = \|x_e-x_{e'}\|$ where $x_e,x_{e'}$ are the feature vectors for images $e,e'$. We normalize the feature vectors to mean zero and unit norm, and we choose $e_0$ as the origin. For our robust criteria, we consider $k=20$ perturbations of the function $f$ defined above, i.e., Problem \eqref{eq:offline_def} corresponds to \(\max_{A\in \I}\min_{i\in[20]}f(A) + \sum_{e\in A\cap \Lambda_i}\eta_e\), where \(\Lambda_i \) is a random set of fixed size with different composition for each $i\in[20]$, and finally, $\eta\sim[0,1]^V$ is a uniform error vector.

We consider two experiments: (\texttt{small}) with $n=3,000$ images, 20 random instances considering $q=6$ and $b=70$, $|\Lambda_i|=500$ and (\texttt{large})  with $n=17,125$ images, 20 random instances $q\in\{10,\ldots,29\}$ parts and budget $b=5$, $|\Lambda_i|=3,000$ (we do not implement \texttt{prev-extG} because of the exorbitant running time). We report the results of the experiments in Figures \ref{fig:experiments} (e)-(h). In Figures \ref{fig:experiments}  (e) and (f), we show the performance profiles for the running time and the number of function calls of all methods, respectively. The y-axis corresponds to the fraction of the instances in which a specific method performs less than a multiplicative factor (x-axis) with respect to the best method. For example, we observe in (f) that in only 20\% of the instances the method \texttt{ext-Greedy} uses less than 1.1 times the number of function calls used by the best method (\texttt{ext-ThGreedy}). We do not present the results of \texttt{prev-extG} due to its difference in magnitude. For \texttt{small}, performance profiles (e) and (f) confirm our theoretical results: \texttt{ext-ThGreedy} is the most likely to use less function calls (f) and running time (e) when the rank is relatively high (in this case $q\cdot b =420$). This contrasts with the performance of \texttt{ext-Greedy} that depends on the rank (chart (f) reflects this). For \texttt{large}, we can see in box-plots (g) and (h) that the results are similar, either in terms of running time or function evaluations. Therefore, when we face a large ground set and a small rank we could choose any algorithm, but we would still prefer \texttt{ext-ThGreedy}, since it has no dependency on the rank.

%% file: online.tex
\section{The Online Case}\label{sec:online_problem}
In this section, we consider the online robust optimization problem (Equation~\ref{eq:regret}) under matroid constraints. We introduce an online bi-criteria algorithm that achieves a sublinear $(1-\epsilon)$-regret while using solution $S^t$ at time $t$ that is a union of $O(\ln\frac{1}{\epsilon})$ independent sets from $\I$.
To start, let us first present definitions and known results that play a key role in this online optimization problem. To avoid any confusion, in the remainder of the section we will denote the dot product between two vectors as $\langle\cdot,\cdot\rangle$.

\subsection{Background}\label{sec:online_background}
\paragraph{Submodular maximization.} Multilinear extension plays a crucial role in designing approximation algorithms for various constrained submodular optimization problems (see Section \ref{sec:cont_greedy} for a list of its useful properties). Vondr\'ak~\citep{vondrak2008optimal} introduced the  \emph{discretized continuous greedy} algorithm that achieves a $1-1/e$ approximate solution for maximizing a single monotone submodular function under matroid constraints (see \citep{feldman2011unified} for the variant of the continuous greedy that we use).  Consider $G$ the multilinear extension of a monotone submodular function $g$. Recall that $\Delta G(y)$ denotes the vector whose $e^{th}$-coordinate is $\Delta_e G(y)$ as defined in \eqref{eq:delta_ML}. At a high level, the discretized continuous greedy algorithm discretizes interval $[0,1]$ into points $\{0,\delta,2\delta,\ldots,1\}$. Starting at $y_0=0$,  for each $\tau\in \{\delta,2\delta,\ldots,1\}$ the algorithm uses an LP to compute the direction $z_\tau = \argmax_{z\in\P(\M)}\langle\Delta G(y_{\tau-\delta}), z\rangle$. Then the algorithm takes a step in the direction of $z_\tau$ by setting
$y_{\tau, e} \gets y_{\tau-\delta, e} + \delta z_{\tau,e} (1 - y_{\tau-\delta, e})$ for  all $e\in V$. Finally, it outputs a set $S$ by rounding the fractional solution $y_1$. We will use this discretized version of the continuous greedy to construct our online algorithm in the following section.

\paragraph{The soft-min function.} Consider a set of $k$ twice differentiable, real-valued functions $\phi_1,\ldots, \phi_k:\RR^n\to\RR$. Let $\phi_{min}$ be the minimum among these functions, i.e., for each point $x$ in the domain, define $\phi_{min}(x):=\min_{i\in[k]} \phi_i(x)$. This function can be approximated by using the so-called {\it soft-min} (or {\it log-sum-exp}) function $H:\RR^n\to\RR$ defined as
\[H(x)=-\frac{1}{\alpha}\ln \sum_{i\in[k]}e^{-\alpha \phi_i(x)},\]
where $\alpha>0$ is a fixed parameter. We now present some of the
key properties of this function in the following lemma.
\begin{lemma}\label{lem:prop_softmin}
For any set of $k$ twice differentiable, real-valued functions $\phi_1,\ldots, \phi_k$, the soft-min function $H$ satisfies the following properties:
\begin{enumerate}
\item Bounds:
\begin{equation}\label{eq:bounds}
\phi_{min}(x)-\frac{\ln k}{\alpha}\leq H(x)\leq \phi_{min}(x).
\end{equation}
\item Gradient: \begin{equation}\nabla H(x)=\sum_{i\in[k]}p_i(x)\nabla
    \phi_i(x),\end{equation} where
  $p_i(x):=e^{-\alpha \phi_i(x)}/\sum_{j\in[k]}e^{-\alpha
    \phi_j(x)}$. Clearly, if $\nabla \phi_i\geq0$ for all $i\in[k]$, then
  $\nabla H\geq 0$.
\item Hessian:
\begin{align}\label{eq:bound_hessianH}
\frac{\partial^2H(x)}{\partial x_{e_1}\partial x_{e_2}}&=\sum_{i\in[k]}p_i(x)\left(-\alpha\frac{\partial \phi_i(x)}{\partial x_{e_1}} \frac{\partial \phi_i(x)}{\partial x_{e_2}}+\frac{\partial^2\phi_i(x)}{\partial x_{e_1}\partial x_{e_2}}\right)+\alpha\nabla_{e_1} H(x) \nabla_{e_2}H(x) 
\end{align}
Moreover, if for all $i\in [k]$ we have $\left|\frac{\partial \phi_i}{\partial x_{e_1}}\right|\leq L_1$, and $\left|\frac{\partial^2 \phi_i}{\partial x_{e_1}\partial x_{e_2}}\right|\leq L_2$, then $\left|\frac{\partial^2 H}{\partial x_{e_1}\partial x_{e_2}}\right|\leq 2\alpha L_1^2 +L_2.$
\item Comparing the average of the $\phi_i$ functions with $H$: given $\alpha>0$ we have
\begin{equation}\label{eq:average_H}
H(x)\leq\sum_{i\in [k]}p_i(x)\phi_i(x)\leq H(x)+\frac{\ln \alpha}{\alpha}+\frac{\ln k}{\alpha}+\frac{k}{\alpha}.
\end{equation}
Therefore, for $\alpha>0$ sufficiently large $\sum_{i\in [k]}p_i(x)\phi_i(x)$ is a good approximation of $H(x)$.
\end{enumerate}
\end{lemma}

For other properties and applications, we refer the interested reader to \citep{calafiore_ghaoui14}. Now, we present a lemma which is used to prove the main result in the online case, Theorem \ref{theorem:online}. 

\begin{lemma}\label{lemma:Taylor_bound_H}
Fix a parameter $\delta>0$. Consider $T$ collections of $k$ twice-differentiable functions, namely $\{\phi_i^1\}_{i\in[k]},\ldots, \{\phi_i^T\}_{i\in[k]}$. Assume $0\leq \phi_i^t(x)\leq 1$ for any $x$ in the domain, for all $t\in[T]$ and $i\in[k]$. Define the corresponding sequence of soft-min functions $H^1,\ldots, H^T$,  with common parameter $\alpha>0$. Then, any two sequences of points $\{x^t\}_{t\in[T]},\{y^t\}_{t\in[T]}\subseteq[0,1]^V$ with $|x^t-y^t|\leq \delta$ satisfy
\[\sum_{t\in[T]}H^t(y^t)-\sum_{t\in[T]}H^t(x^t)\geq\sum_{t\in[T]}\langle\nabla H^t(x^t), y^t-x^t\rangle- O\big(Tn^2\delta^2\alpha\big).\]
\end{lemma}

For a proof of these lemmas, we refer to the Appendix. 

\subsection{Online Algorithm and Analysis}\label{sec:online_analysis}

Turning our attention to the online robust optimization problem \eqref{eq:regret}, we are immediately faced with two challenges. First, we need  to find a direction $z_t$ that is good for  all $k$ submodular functions in an online fashion. In the offline case, we used the function $g(S)=\frac{1}{k}\cdot\sum_{i=1}^k\min\{f_i(S),\gamma\}$ and its multilinear extension to find such a direction (see Section \ref{sec:cont_greedy}). To resolve this issue, we use a \emph{soft-min} function that converts robust optimization over $k$ functions into optimizing of a single function.
Secondly, robust optimization leads to non-convex and non-smooth optimization combined with online arrival of such submodular functions. To deal with this, we use the  Follow-the-Perturbed-Leader (FPL)  online  algorithm introduced by \citet{kalai_etal05}.

\nhdelete{For any $i\in[k]$ and $t\in[T]$, we denote by $F_i^t$ the multilinear extension of $f_i^t$. Also, similarly as in \eqref{eq:delta_def} we define $\Delta_e F_i^t(y)$ for all $e\in V$.}
For any collection of monotone submodular functions $\{f_i^t\}_{i\in [k]}$ played by the adversary, we define the \emph{soft-min} function with respect to the corresponding multilinear extensions $\{F_i^t\}_{i\in [k]}$ as 
\(H^t(y) := -\frac{1}{\alpha}\ln \sum_{i\in[k]}e^{-\alpha F_i^t(y)},\)
where $\alpha>0$ is a suitable parameter. Recall we assume functions $f_i^t$ taking values in $[0,1]$, then their multilinear extensions $F_i^t$ also take values in $[0,1]$. The following properties of the soft-min function as defined in the previous section are easy to verify and crucial for our result. 
\vspace{1em}
\begin{enumerate}
\item Approximation:
\(\min_{i\in[k]} F_i^t(y) - \dfrac{\ln k}{\alpha} \leq H^t(y) \leq \min_{i\in[k]} F_i^t(y).\)
\vspace{1em}
\item Gradient: \(\nabla H^t(y) = \sum_{i\in[k]} p^t_i(y)\nabla F_i^t(y),\)
where $ p_i^t(y)\propto e^{-\alpha F_i^t(y)}\text{ for all } i\in[k]$.
\vspace{1em}
\end{enumerate}
Note that as $\alpha$ increases, the soft-min function $H^t$ becomes a better approximation of  $\min_{i\in[k]}\{F_i^t\}$, however, its smoothness degrades (see Property \eqref{eq:bound_hessianH} in Section \ref{sec:online_background}). 
On the other hand, the second property shows that the gradient of the soft-min function is a convex combination of the gradients of the multilinear extensions, which allows us to optimize all the functions at the same time. Indeed, define $\Delta_e H^t(y):= \sum_{i\in[k]} p^t_i(y)\Delta_e F_i^t(y)=(1-y_e)\nabla_e H^t(y).$
At each stage $t\in [T]$, we use the information from the gradients previously observed, in particular, $\{\Delta H^1,\cdots,\Delta H^{t-1}\}$ to decide the set $S^t$. 
To deal with adversarial input functions, we use the FPL algorithm~\citep{kalai_etal05} and the following guarantee about the algorithm.
\begin{theorem}[\citet{kalai_etal05}]\label{theo:regret_FPL2}
  Let $s_1,\ldots, s_T\in\mathcal{S}$ be a sequence of rewards. The FPL algorithm 6 (in the Appendix) with parameter $\eta\leq 1$ outputs decisions $d_1,\ldots,d_T$ with
  regret
\begin{equation*}
\max_{d\in\D}\sum_{t\in[T]}\langle s_t, d\rangle -  \EE\left[\sum_{t\in[T]}  \langle s_t, d_t\rangle\right] \leq O\left(\text{poly}(n)\left(\eta T+\frac{1}{T\eta}\right)\right).
\end{equation*}
\end{theorem}
For completeness, we include the original setup and the algorithm in the Appendix. 

Our online algorithm works as follows: first, given $0<\epsilon<1$ we denote $\ell := \lceil \ln \frac{1}{\epsilon} \rceil$. We consider the following discretization indexed by $\tau\in\{0,\delta,2\delta,\ldots,\ell\}$ and construct fractional solutions $y^t_{\tau}$ for each iteration $t$ and discretization index $\tau$. At each iteration $t$, ideally we would like to construct $\{y^t_{\tau}\}_{\tau=0}^\ell$ by running the continuous greedy algorithm using the soft-min function $H^t$ and then play $S^t$ using these fractional solutions. But in the online model, function $H^t$ is revealed only after playing set $S^t$. To remedy this, we aim to construct $y^t_{\tau}$ using FPL algorithm based on gradients $\{\nabla H^j\}_{j=1}^{t-1}$ obtained from previous iterations. Thus we have multiple FPL instances, one for each discretization parameter, being run by the algorithm.
Finally, at the end of iteration $t$, we have a fractional vector $y^t_\ell$ which belongs to $\ell \cdot \P(\M)\cap [0,1]^V$ and therefore can be written, fractionally, as a union of $\ell$ independent sets using the matroid union theorem~\citep{schrijver2003combinatorial}.

We round the fractional solution $y_{\ell}^t$ using the randomized swap rounding (or randomized pipage rounding) proposed by \citet{chekuri2010dependent} for matroid $\M_{\ell}$ to obtain the set $S^t$ to be played at time $t$. The following theorem from \citep{chekuri2010dependent} gives the necessary property of the randomized swap rounding that we use.
\begin{theorem}[Theorem II.1, \citep{chekuri2010dependent}]\label{thm:swap-rounding2}
Let $f$ be a monotone submodular function and $F$ be its multilinear extension. Let $ x\in \P(\M')$ be a point in the polytope of matroid $\M'$ and $S'$ a random
independent set obtained from it by randomized swap rounding. Then, $\EE[f(S')]\geq F(x)$.
\end{theorem}

Below, we formalize the details in Algorithm \ref{alg:OnlineSoftMin} (observe that $\ell/\delta \in \ZZ_+$). 
\begin{algorithm}[h!]
\caption{OnlineSoftMin algorithm}\label{alg:OnlineSoftMin}
\begin{algorithmic}[1]
\Require  learning parameter $\eta>0$, $\epsilon>0$, $\alpha= nT$, discretization $\delta= n^{-1}T^{-1}$, and $\ell = \lceil  \ln \frac{1}{\epsilon}\rceil$.
\Ensure sequence of sets $S_1,\ldots, S_T$.
\State Sample $q\sim[0,1/\eta]^V$
\For {$t=1$ to $T$}
\State $y_0^t=0$
\For {$\tau \in \{\delta,2\delta,\ldots, \ell \}$}
	 \State {\bf Compute} 	 \( z^t_\tau = \argmax_{z\in \P(\M)}\left\langle\sum_{j=1}^{t-1} \Delta H^{j}(y^j_{\tau-\delta})+q, z\right\rangle.\)
	\State {\bf Update} For each $e\in V$,  
	\(
	y_{\tau,e}^t=y_{\tau-\delta,e}^t+\delta (1-y_{\tau-\delta,e}^t) z_{\tau,e}^t.\)
\EndFor
\State{\bf Play} $S^t \leftarrow \ \text{Rounding}\left(y_{\ell}^t\right)$. {\bf Receive} and {\bf observe} new collection $\{f_i^t\}_{i\in[k]}$.
\EndFor
\end{algorithmic}
\end{algorithm}
In order to get sub-linear regret for the FPL algorithm 6, \cite{kalai_etal05} assume a couple of conditions on the problem (see Appendix). Similarly, for our online model we need to consider the following for any $t\in[T]$:
\begin{enumerate}
\item bounded diameter of $\P(\M)$, i.e., for all $z,z'\in \P(\M)$, $\| z-z'\|_1\leq D$;
\item for all $y,z\in \P(\M)$, we require $\left|\langle z,\Delta H^t(y)\rangle\right|\leq L$;
\item for all $y\in \P(\M)$, we require $\|\Delta H^t(y)\|_1\leq A$,
\end{enumerate}

Now, we give a complete proof of Theorem \ref{theorem:online} for any given learning parameter $\eta>0$, but the final result follows with $\eta= \sqrt{D/LAT}$ and assuming $L\leq n$, $A\leq n$ and $D\leq \sqrt{n}$, which gives a $O(n^{5/4})$ dependency on the dimension in the regret.

\begin{proof}[Proof of Theorem \ref{theorem:online}.]
Consider the sequence of multilinear extensions
  $\{F_i^1\}_{i\in[k]}, \ldots, \{F_i^T\}_{i\in[k]}$ derived from the
  monotone submodular functions $f_i^t$ obtained during the dynamic
  process. Since $f_i^t$'s have value in $[0,1]$, we have $0\leq F_i^t(y)\leq1$ for any $y\in[0,1]^V$ and $i \in [k]$.
  Consider the corresponding soft-min functions $H^t$ for collection $\{F_i^t\}_{i\in[k]}$ with $\alpha = nT$ for all $t\in[T]$. Denote $\ell = \lceil \ln \frac{1}{\epsilon} \rceil$ and fix
  $\tau\in \{\delta,2\delta,\ldots, \ell\}$ with
  $\delta=n^{-1}T^{-1}$. According to the update in Algorithm
  \ref{alg:OnlineSoftMin}, $\{y^t_{\tau}\}_{t\in[T]}$ and
  $\{y^t_{\tau-\delta}\}_{t\in[T]}$ satisfy conditions of Lemma
  \ref{lemma:Taylor_bound_H}. Thus, we obtain
\begin{equation*}
\sum_{t\in[T]}H^t(y^t_\tau)-H^t(y^t_{\tau-\delta})\geq \sum_{t\in[T]}\left\langle\nabla H^t(y^t_{\tau-\delta}), y^t_\tau-y^t_{\tau-\delta}\right\rangle - O\Big(Tn^2\delta^2\alpha\Big).
\end{equation*}

Then, since the update is $y_{\tau,e}^t=y_{\tau-\delta,e}^t+\delta (1-y_{\tau-\delta,e}^t) z_{\tau,e}^t$, we get
\begin{align}
\sum_{t\in[T]}H^t(y^t_\tau)-H^t(y^t_{\tau-\delta})\notag&\geq \delta\sum_{t\in[T]}\sum_{e\in V} \nabla_e H^t(y^t_{\tau-\delta})(1-y_{\tau-\delta,e}^t) z_{\tau,e}^t - O\Big(Tn^2\delta^2\alpha\Big) \notag\\
&= \delta\sum_{t\in[T]} \left\langle\Delta H^t(y^t_{\tau-\delta}), z_{\tau}^t\right\rangle- O\Big(Tn^2\delta^2\alpha\Big). \label{eq:update2}
\end{align}

Observe that an FPL algorithm is implemented for each $\tau$, so we can state a regret bound for each $\tau$ by using Theorem \ref{theo:regret_FPL2}. Specifically,
\begin{equation*}
\EE\left[\sum_{t\in[T]}\left\langle\Delta H^t(y^t_{\tau-\delta}), z^t_\tau\right\rangle\right]\geq\max_{z\in \P(\M)} \EE\left[\sum_{t\in[T]}\left\langle\Delta H^t(y^t_{\tau-\delta}), z\right\rangle\right] -\text{ $R_\eta$},
\end{equation*}
where $R_\eta=\eta LAT+\frac{D}{\eta}$ is the regret guarantee for a given $\eta>0$. By taking expectation in \eqref{eq:update2} and using the regret bound we just mentioned, we obtain
\begin{align}\label{eq:ineqH}
\EE\Bigg[\sum_{t\in[T]}H^t(y^t_\tau)-H^t(y^t_{\tau-\delta})\Bigg]&\geq \delta\left(\max_{z\in \P(\M)}\EE\left[\sum_{t\in[T]}\left\langle\Delta H^t(y^t_{\tau-\delta}), z\right\rangle\right] \right)- \delta R_\eta  - O(Tn^2\delta^2\alpha)\notag \\
&\geq \delta \EE\left(\sum_{t\in[T]} \left[H^t(x^*)-
  \sum_{i\in[k]}p^t_i(y^t_{\tau-\delta})F^t_i(y^t_{\tau-\delta})
  \right]\right) \notag \\
  & \qquad -\delta R_\eta- O(Tn^2\delta^2\alpha),
 \end{align}
 where $x^*=\one_{S^*}$ is the indicator vector of the true optimum $S^*$ for $\max_{S\in \I}\sum_{t\in[T]}\min_{i\in[k]}f_i^t(S)$. Observe that \eqref{eq:ineqH} follows from monotonicity and submodularity of each $f_i^t$, specifically we know that
 \begin{align*}
 \left\langle\Delta H^t(y), z\right\rangle=\sum_{i\in[k]}p^t_i(y) \left\langle\Delta F_i^t(y), z\right\rangle&\geq \sum_{i\in[k]}p^t_i(y) F^t_i(x^*) - \sum_{i\in[k]}p^t_i(y) F^t_i(y) \qquad  \text{(eq. \eqref{eq:grad})}\\
  &\geq  F^t_{min}(x^*) - \sum_{i\in[k]}p^t_i(y) F^t_i(y)\\
 &\geq H^t(x^*) - \sum_{i\in[k]}p^t_i(y) F^t_i(y).
 \end{align*}
 By applying property \eqref{eq:average_H} of the soft-min in expression \eqref{eq:ineqH} we get
\begin{align}\label{eq:H_iterate}
\EE \left[\sum_{t\in[T]}H^t(y^t_\tau)-H^t(y^t_{\tau-\delta})\right]\geq \notag& \delta \EE\left(\sum_{t\in[T]}H^t(x^*)- H^t(y_{\tau-\delta}^t)\right)-\delta R_\eta-  O(Tn^2\delta^2\alpha) \notag\\
&-\delta T\left(\frac{\ln \alpha}{\alpha}+\frac{\ln k}{\alpha}+\frac{k}{\alpha}\right),
 \end{align}
 Given the choice of $\alpha$ and $\delta$, the last two terms in the right-hand side of inequality \eqref{eq:H_iterate} are small compared to $R_\eta$, so by re-arranging terms we can state the following
 \begin{equation*}\sum_{t\in[T]} H^t(x^*)- \EE\left[\sum_{t\in[T]}H^t(y^t_\tau)\right] \leq (1-\delta)\left(\sum_{t\in[T]}H^t(x^*)- \EE\left[\sum_{t\in[T]}H^t(y^t_{\tau-\delta})\right]\right)+ 2\delta R_\eta\end{equation*}
 By iterating $\frac{\ell}{\delta}$ times in $\tau$, we get
 \begin{align*}
\sum_{t\in[T]}H^t(x^*)- \EE\left[\sum_{t\in[T]}H^t(y^t_{\ell})\right]&\leq (1-\delta)^{\frac{\ell}{\delta}}\left(\sum_{t\in[T]}H^t(x^*)-\sum_{t\in[T]}H^t(y^t_0)\right)+O\left((1-(1-\delta)^{\frac{\ell}{\delta}})R_\eta \right) \\
&\leq \epsilon\left[\sum_{t\in[T]}H^t(x^*)+\ln k\right]+O\left( (1-\epsilon)R_\eta\right),
 \end{align*}
where in the last inequality we used $(1-\delta)\leq e^{-\delta}$ and $\ell = \lceil  \ln \frac{1}{\epsilon}\rceil$. Given that the term $\epsilon\ln k$ is small (for $\epsilon$ sufficiently small) we can bound it by $O(R_\eta)$. Since $\alpha$ is sufficiently large, we can apply the approximation property of soft-min function to obtain the following regret bound
\begin{equation*}(1-\epsilon)\cdot\sum_{t\in[T]}\min_{i\in[k]}F_i^t(x^*)-\EE\left[\sum_{t\in[T]}\min_{i\in[k]}F_i^t\left(y^t_{\ell}\right)\right]\leq O\left((1-\epsilon)R_\eta\right).\end{equation*}
Since we are doing randomized swap rounding (or randomized pipage rounding) on each $y^t_{\ell}$, Theorem \ref{thm:swap-rounding2} shows that there is  a random set $S^t$ that is independent in $\M_\ell$ (i.e., $S^t$ is the union of at most $\ell$ independent sets in $\I$)  such that
$\EE\left[f^t_i(S^t)\right]\geq F_i^t\left(y^t_{\ell}\right)$ for all $t\in[T]$ and $i\in[k]$. Thus, we finally obtain
\begin{equation*}(1-\epsilon)\cdot \max_{S\in\I}\sum_{t\in[T]}\min_{i\in[k]}f^t_i(S)-\sum_{t\in[T]}\min_{i\in[k]}\EE\left[f^t_i(S^t)\right] \leq  O\left((1-\epsilon)R_\eta\right).\end{equation*}
\end{proof}
\begin{observation}
Theorem \ref{theorem:online} could be easily extended to an adaptive adversary by sampling in each stage $t\in[T]$ a different perturbation $q_t\sim[0,1/\eta]^V$ as shown in \citep{kalai_etal05}.
\end{observation} 
Note that the guarantee of Theorem~\ref{theorem:online} holds with respect to the minimum of $\EE[f_i^t(S^t)]$, as opposed to the guarantee of Theorem~\ref{theorem1:offline} that directly bounds the minimum of $f_i(S)$. Because of this, the online algorithm needs only $\lceil\ln \frac{1}{\epsilon}\rceil$ independent sets, compared to the offline solution which needs $\lceil\log \frac{k}{\epsilon}\rceil$ independent sets. It might seem more appealing to define the regret with respect to the expected value of the minimum function, but at the same time it becomes technically more challenging for several reasons. If one wants to apply a similar technique than the offline algorithm, it is not clear how to \emph{dynamically estimate} the optimal value of the online model, while considering simultaneously the behavior of the adversary and the non-smoothness of the minimum function.

%% file: extensions.tex
\section{Extensions}\label{sec:extensions}
In this section, we consider other classes of combinatorial constraints for the offline robust model. Since the extended algorithm \texttt{ext-$\A$} considers a general algorithm $\A$ for submodular maximization, we can expand our results to other constraints such as knapsack constraints and multiple matroids. Similar results can be obtained in the online model as long as the polytope is downward-closed.

\subsection{Knapsack constraint}
Consider a knapsack constraint $\K=\{S\subseteq [n]: \ \sum_{e\in S}c_e\leq 1\}$, where $c_e>0$ for all $e\in[n]$. Our interest is to solve the following robust problem
\begin{equation}\label{eq:knapsack}
\max_{S\in\K}\min_{i\in[k]} f_i (S)
\end{equation}

\begin{corollary}\label{cor:knapsack}
For Problem \eqref{eq:knapsack}, there is a polynomial time algorithm that returns a set $S^{\alg}$, such that for all $i\in[k]$,
for a given $0<\epsilon<1$,
\[f_i(S^{\alg})\geq (1-\epsilon) \cdot \max_{S\in\K}\min_{j\in[k]} f_j (S),\]
and $\sum_{e\in S^{\alg}}c_e \leq \ell $ for $\ell = { O(\ln\frac{k}{\epsilon})}$. Moreover, $S^{\alg}$ can be covered by at most $\ell$ sets in $\K$.
\end{corollary}

Following the idea of the general extended algorithm \texttt{ext-$\A$}, we design an extended version of the ``bang-per-buck'' greedy algorithm. We formalize this procedure in Algorithm \ref{alg:ext_knapsack} below. Even though the standard ``bang-per-buck'' greedy algorithm does not provide any approximation factor, if we relax the knapsack constraint to be $\sum_{e\in S}c_e\leq 2$, then the algorithm gives a $1-1/e$ factor. There are other approaches to avoid this relaxation, see e.g. \citep{sviridenko_04}.
\begin{algorithm}[h!]
\small
\caption{Extended ``Bang-per-Buck'' Algorithm for Knapsack Constraints}\label{alg:ext_knapsack}
\begin{algorithmic}[1]
\Require $\ell\geq 1$, monotone submodular function $g:2^{V} \rightarrow \RR_+$, knapsack constraint $\K$.
\Ensure sets $S_1,\ldots, S_\ell\in \K$.

\For {$\tau=1, \dots, \ell$}

\State $S_\tau\gets \emptyset$

\While {$V\neq \emptyset$}

	 \State {\bf Compute} \(e^* = \argmax_{e\in V}\left\{ \frac{g(\cup_{j=1}^{\tau} S_j + e) - g(\cup_{j=1}^{\tau} S_j )}{c_e}\right\}. \)
	
	\If {$\sum_{e\in S^\tau}c_e + c_{e^*} \leq 2$}   $ S_\tau\gets S_\tau+e^*.$ \EndIf
	
	\State $V \leftarrow V - e^*$

\EndWhile
\State {\bf Restart} ground set $V$.
\EndFor
\end{algorithmic}
\end{algorithm}

Given a monotone submodular function $g:2^V\to\RR_+$, Algorithm \ref{alg:ext_knapsack} produces a set $S^{\alg}=S_1\cup\cdots\cup S_\ell$ such that
\(g(S^{\alg})\geq \left( 1 - \frac{1}{e^\ell}\right)\cdot \max_{S\in\K} g(S).\)
Therefore, Corollary \ref{cor:knapsack} can be easily proved by defining $g$ in the same way as in Theorem \ref{theorem1:offline}, and running Algorithm \ref{alg:ext_knapsack} on $g$ with \(\ell = { O(\ln\frac{k}{\epsilon})}\).

\subsection{Multiple matroid constraints}
Consider a family of $r$ matroids $\M_j=(V,\I_j)$ for $j\in[r]$. Our interest is to solve the following robust problem
\begin{equation}\label{eq:mult_matroids}
\max_{S\in\bigcap_{j=1}^r\I_j}\min_{i\in[k]} f_i (S)
\end{equation}

\begin{corollary}\label{cor:mult_matroids}
For Problem \eqref{eq:mult_matroids}, there is a polynomial time algorithm that returns a set $S^{\alg}$, such that for all $i\in[k]$,
for a given $0<\epsilon<1$,
\[f_i(S^{\alg})\geq (1-\epsilon) \cdot \max_{S\in\bigcap_{j=1}^r\I_j}\min_{i\in[k]} f_i (S),\]
where $S^{\alg}$ is the union of $O(\log \frac{k}{\epsilon}/\log \frac{r+1}{r} )$ independent sets in $\I$.
\end{corollary}

\citet{fisher1978analysis} proved that the standard greedy algorithm gives a $1/(1+r)$ approximation for Problem \eqref{eq:mult_matroids} when $k=1$. Therefore, we can adapt Algorithm \ref{alg:ext_greedy} to produce a set $S^{\alg}=S_1\cup\cdots\cup S_\ell$ such that
\[f(S^{\alg})\geq \left( 1 - \left(\frac{r}{r+1}\right)^\ell \right) \cdot \max_{S\in\bigcap_{j=1}^r\I_j} f(S).\]
Then, Corollary \ref{cor:mult_matroids} can be proved similarly to Theorem \ref{theorem1:offline} by choosing \(\ell = O(\log \frac{k}{\epsilon}/\log \frac{r+1}{r} )\)

\subsection{Distributionally robust over polyhedral sets}
Let \(\mathcal Q \subseteq \Delta(k)\) be a polyhedral set, where \(\Delta(k)\) is the probability simplex on \(k\) elements. For \(q \in \mathcal Q\), denote \(f_q \coloneqq q_1f_1+\cdots+q_kf_k\), which is also monotone and submodular. Given a matroid $\M=(V,\I)$, our interest is to solve the following distributionally robust problem
\begin{equation}\label{eq:dist_robust}
\max_{S\in\I}\min_{q\in\Qe} f_q (S)
\end{equation}
Denote by $\text{Vert}(\Qe)$ the set of extreme points of $\Qe$, which is finite since $\Qe$ is polyhedral. Then, Problem \eqref{eq:dist_robust} is equivalent to $\max_{S\in\I}\min_{q\in\text{Vert}(\Qe)}f_q(S)$. Then, we can easily derive Corollary \ref{cor:dist_robust} (below) by applying Theorem \ref{theorem1:offline} in the equivalent problem. Note that when $\Qe$ is the simplex we get the original Theorem \ref{theorem1:offline}.

\begin{corollary}\label{cor:dist_robust}
For Problem \eqref{eq:dist_robust}, there is a polynomial time algorithm that returns a set $S^{\alg}$, such that for all $i\in[k]$,
for a given $0<\epsilon<1$,
\[f_i(S^{\alg})\geq (1-\epsilon) \cdot \max_{S\in \I} \min_{q\in\Qe}f_q(S),\]
with $S^{\alg}=S_1\cup \dots\cup S_\ell$ for $\ell = { O(\log\frac{|\text{Vert}(\Qe)|}{\epsilon})}$ and $S_1,\dots,S_\ell\in \I$.
\end{corollary}

%% file: appendix.tex
\section{Remaining Proofs}
\label{sec:remaining_proofs}
\begin{proof}[Proof of Corollary \ref{cor:threshold_greedy}.]
Denote by $r$ the rank of matroid $\M$. Let $S^*=\{e_1^*,\ldots, e_r^*\}$ and  $S=\{e_1,\ldots,e_r\}$ be the optimal set and the set obtained with the algorithm, respectively. W.l.o.g, we can assume that $S^*$ and $S_G$ are both basis in $\M$, so there exists a bijection $\phi$ such that $\phi(e_i) = e_i^*$ for all $i\in[r]$. Denote by $S_{i-1}=\{e_1,\ldots,e_{i-1}\}$ the set of elements after iteration $i-1$. Observe that if  $e_i$ is the next element chosen by the algorithm and the current threshold value is $w$, then we get the inequalities
\[g_{S_{i-1}}(x) = \left\{\begin{matrix} \geq  w & \text{if} \ x=e_i \\ \leq w/(1-\delta) & \text{if} \ x\in V \ \text{s.t.}  \ S_{i-1}+x \in \I\end{matrix}\right.\]
This imples that $g_{S_{i-1}}(e_i)\geq (1-\delta)f_{S_{i-1}}(x)$ for all $x\in V\backslash S_{i-1}$ such that $S_{i-1} + x\in \I$. In particular for $x = e_i^*$ we have then
\[(1-\delta)g_{S_{i-1}}(e_i^*)\leq g(S_i) - g(S_{i-1}).\]
 On the other hand, if we apply submodularity twice we get
\[g(S^*)-g(S)\leq \sum_{i=1}^Kg_{S}(e_i^*)\leq \sum_{i=1}^Kg_{S_{i-1}}(e_i^*)\]
Using the previous two inequalities we get
\[(1-\delta)[g(S^*)-g(S)] \leq  \sum_{i=1}^K  g(S_i) - g(S_{i-1}) = g(S).\]
So we finally obtain
\[g(S)\geq \left(1- \frac{1}{2-\delta} \right) \cdot g(S^*). \]
\end{proof}

\vspace{1em}

\begin{proof}[Proof of Lemma \ref{lem:prop_softmin}.]
We will just prove properties 1 and 4, since the rest is a straightforward calculation.
\begin{enumerate}
\item[1.] First, for all $i\in[k]$ we have $e^{-\alpha \phi_i(x)}\leq  e^{-\alpha \phi_{min}(x)}$. Thus,
\begin{equation*}H(x)=-\frac{1}{\alpha}\ln\sum_{i\in[k]}e^{-\alpha \phi_i(x)}\geq -\frac{1}{\alpha}\ln \left(ke^{-\alpha \phi_{min}(x)}\right)= \phi_{min}(x)-\frac{\ln k}{\alpha}\end{equation*}
On the other hand, $\sum_{i\in[k]}e^{-\alpha \phi_i(x)}\geq e^{-\alpha \phi_{min}(x)}$. Hence,
\[H(x)\leq -\frac{1}{\alpha}\ln\left(e^{-\alpha \phi_{min}(x)}\right)=\phi_{min}(x).\]

\item[4.] Let us consider sets $A_1=\{i\in[k]: \ \phi_i(x)\leq \phi_{min}(x)+(\ln \alpha)/\alpha\}$ and $A_2=\{i\in[k]: \ \phi_i(x)> \phi_{min}(x)+(\ln \alpha)/\alpha\}$. Our intuitive argument is the following: when $\alpha$ is sufficiently large, those $p_i(x)$'s with $i\in A_2$ are exponentially small, and $p_i(x)$'s with $i\in A_1$ go to a uniform distribution over elements in $A_1$. First, observe that for each $i\in A_2$ we have
 \begin{equation*}p_i(x)=\frac{e^{-\alpha \phi_i(x)}}{\sum_{i\in[k]}e^{-\alpha \phi_i(x)}}<\frac{e^{-\alpha[\phi_{min}(x)+(\ln \alpha)/\alpha]}}{e^{-\alpha \phi_{min}(x)}}=\frac{1}{\alpha},\end{equation*}
 so $\sum_{i\in A_2}p_i(x)\phi_i(x)\leq \frac{k}{\alpha}.$
 On the other hand, for any $i\in A_1$ we have
 \begin{equation*}\sum_{i\in A_1}p_i(x)\phi_i(x)\leq \left(\phi_{min}(x)+\frac{\ln \alpha}{\alpha}\right)\sum_{i\in A_1}p_i(x) \leq H(x)+\frac{\ln \alpha}{\alpha}+\frac{\ln k}{\alpha}\end{equation*}
 where in the last inequality we used the approximation property of the soft-min function. Therefore,
 \[\sum_{i\in [k]}p_i(x)\phi_i(x)\leq H(x)+\frac{\ln \alpha}{\alpha}+\frac{\ln k}{\alpha}+ \frac{k}{\alpha}.\]
 Finally, the other inequality is clear since $\sum_{i\in [k]}p_i(x)\phi_i(x)\geq \phi_{min}(x)\geq H(x)$. 
\end{enumerate}
\end{proof}

\vspace{1em}
\begin{proof}[Proof of Lemma \ref{lemma:Taylor_bound_H}.]
For every $t\in[T]$ define a matroid $\M_t=(V\times\{t\},\I\times\{t\})=(V_t,\I_t)$. Given this, the union matroid is given by a ground set $V^{[T]}=\bigcup_{t=1}^TV_t$, and independent set family $\I^{[T]}=\{S\subseteq V^{1:T}: \ S\cap V_t\in \I_t\}$. Define $\HH(X):=\sum_{t\in[T]} H^t(x^t)$ for any matrix $X\in \P(\M)^T$, where $x^t$ denotes the $t$-th column of $X$. Clearly, $\nabla_{(e,t)} \HH(X)=\nabla_e H^t(x^t)$. Moreover, the Hessian corresponds to
\[\nabla^2_{(e_1,t),(e_2,s)} \HH(X)=\left\{\begin{matrix}0 & \ \text{if} \ t\neq s\\ \nabla^2_{e_1,e_2}\HH^t(x^t) & \ \text{if} \ t=s\end{matrix}\right.\]
Consider any $X,Y\in \P(\M)^T$ with $|y^t_e-x^t_e|\leq\delta$. Therefore, a Taylor's expansion of $\HH$ gives
\begin{equation*}\HH(Y)=\HH(X)+\nabla \HH(X)(Y-X)+\frac{1}{2}(Y-X)^\top \nabla^2\HH(\xi) (Y-X)\end{equation*}
where $\xi$ is on the line between $X$ and $Y$. If we expand the previous expression we obtain
\begin{equation*}\HH(Y)-\HH(X)=\sum_{t\in[T]}\left\langle\nabla H^t(x^t),y^t-x^t\right\rangle+\frac{1}{2}\sum_{e_1,e_2\in V}\sum_{t\in[T]}(y^t_{e_1}-x^t_{e_1})\nabla^2_{e_1,e_2}H^t(\xi)(y^t_{e_2}-x^t_{e_2})\end{equation*}
Finally, by using property 3 in Lemma \ref{lem:prop_softmin} and by bounding the Hessian (and ussing the fact that $\phi_i^t(x)\in[0,1]$) we get
\begin{equation*} \HH(Y)-\HH(X)\geq\sum_{t\in[T]}\left\langle\nabla H^t(x^t),y^t-x^t\right\rangle-O(Tn^2\delta^2\alpha),\end{equation*}
which is equivalent to
\begin{equation*}\sum_{t\in[T]}H^t(y^t)-\sum_{t\in[T]}H^t(x^t)\geq\sum_{t\in[T]}\left\langle\nabla H^t(x^t),y^t-x^t\right\rangle-O(Tn^2\delta^2\alpha). \end{equation*}
\end{proof}

\section{Follow-the-Perturbed-Leader Algorithm}\label{sec:FPL}
In this section, we briefly recall the well-known
Follow-the-Perturbed-Leader (FPL) algorithm introduced in
\citep{kalai_etal05} and used in many online optimization problems (see
e.g.,~\citep{rakhlin2009lecture}). The classical online learning
framework is as follows: Consider a dynamic process over $T$ time
steps. In each stage $t\in[T]$, a decision-maker has to choose a
point $d_t \in \D$ from a fixed (possibly infinite) set of actions
$\D\subseteq\RR^n$, then an adversary chooses a vector $s_t$ from a
set $\mathcal{S}$. Finally, the player observes vector $s_t$ and
receives reward \(\langle s_t, d_t\rangle\), and the process continues. The goal of
the player is to maximize the total reward
$\sum_{t\in[T]} \langle s_t, d_t\rangle$, and we compare her performance with
respect to the best single action picked in hindsight, i.e.,
$\max_{d\in\D}\sum_{t\in[T]} \langle s_t\cdot d\rangle$.  This performance with respect
to the best single action in hindsight is called (expected)
\emph{regret}, formally:
\[\text{\bf Regret}(T)=\max_{d\in\D}\sum_{t\in[T]} \langle s_t, d\rangle - \EE\left[\sum_{t\in[T]}  \langle s_t, d_t\rangle\right].\]
\citet{kalai_etal05} showed that even if one has only
access to a linear programming oracle for \(\D\), i.e., we can efficiently solve
$\max_{d\in\D}\langle s, d\rangle$ for any $s\in\mathcal{S}$, then the FPL
algorithm~\ref{alg:FPL} achieves sub-linear regret, specifically
$O(\sqrt{T})$.

In order to state the main result in \citep{kalai_etal05}, we need the
following. We assume that the decision set \(\D\) has diameter at most
\(D\), i.e., for all $d,d'\in\D$, $\| d-d'\|_1\leq D$. Further, for
all $d\in\D$ and $s\in\mathcal{S}$ we assume that the absolute reward is
bounded by \(L\), i.e., $\left| \langle d, s\rangle\right|\leq L$ and that the
\(\ell_1\)-norm of the reward vectors is bounded by \(A\), i.e., for all
$s\in\mathcal{S}$, $\|s\|_1\leq A$.
\begin{theorem}[\citet{kalai_etal05}]\label{theo:regret_FPL}
  Let $s_1,\ldots, s_T\in\mathcal{S}$ be a sequence of rewards. Running the
  FPL algorithm \ref{alg:FPL} with parameter $\eta\leq 1$ ensures
  regret
\[\text{\bf Regret}(T)\leq \eta LAT+\frac{D}{\eta}.\]
Moreover, if we choose $\eta=\sqrt{D/LAT}$, then $\text{\bf
  Regret}(T)\leq 2\sqrt{DLAT} = O(\sqrt{T})$.
\end{theorem}

 \begin{algorithm}[h!]
\caption{Follow-the-Perturbed-Leader (FPL), \citep{kalai_etal05}}\label{alg:FPL}
\begin{algorithmic}[1]
\small
\Require Parameter \(\eta > 0\)
\Ensure Sequence of decisions \(d_1, \dots, d_T\)
\State Sample $q\sim[0,1/\eta]^n$.
\For {$t=1$ to $T$}
\State \textbf{Play} $d_t=\argmax_{d\in\D}\langle \sum_{j=1}^{t-1}s_j + q, d\rangle$.
\EndFor
\end{algorithmic}
\end{algorithm}

\section{Extended Stochastic Greedy for Partition Matroid}\label{sec:E-StochG}
Consider a partition \(\{P_1,\ldots,P_q\}\) on ground set $V$ with $n_j := |P_j|$ for all $j\in[q]$ and a family of feasible sets $\I = \{S\subseteq V: \ |S\cap P_j|\leq k_j \ \forall j\in[q]\}$ which for a matroid $\M=(V,\I)$. We can construct a heuristic based on the stochastic greedy algorithm \citep{mirzasoleiman_etal15} and adapted to partition constraints (Algorithm \ref{alg:ext_partition_stochastic_greedy}): given $\epsilon'>0$, in each round it uniformly samples $\frac{n_j}{b}\log \frac{1}{\epsilon'}$ elements from each part $R_j\sim P_j$, where $n_j:=|P_j|$. And then, it obtains the element with the largest marginal value among elements in $\cup_{j\in[q]}R_j$. 

Even though, we are not able to state any provable guarantee, we use Algorithm \ref{alg:ext_partition_stochastic_greedy} as inner loop for solving the robust problem \eqref{eq:offline_def} with $\ell = \lceil \log \frac{2k}{\epsilon}\rceil$.

\begin{algorithm}[htbp]
\small
\caption{ Extended Stochastic-Greedy for Partition Matroid}\label{alg:ext_partition_stochastic_greedy}
\begin{algorithmic}[1]
\renewcommand{\algorithmicrequire}{\textbf{Input:}}
\renewcommand{\algorithmicensure}{\textbf{Output:}}
\Require $\ell\geq 1$, monotone submodular function $g:2^{V} \rightarrow \RR_+$, partition matroid $\M=(V,\I)$, $\epsilon'>0$.
\Ensure sets $S_1,\ldots, S_\ell\in \I$.
\For {$\tau=1, \dots, \ell$}
\State $S_\tau\leftarrow\emptyset$
\While {$S_\tau$ is not basis in $\M$}
\State For each $j\in[q]$, uniformly sample $R_j\sim P_j\backslash S_\tau$ with $\frac{n_j}{k_j}\log \frac{1}{\epsilon'}$ elements.
\State $e^*\leftarrow \argmax_{e\in R_1\cup\cdots R_r} \left\{g_{\cup_{j=1}^{\tau} S_j}( e)\right\}$.
\State $S_\tau\leftarrow S_\tau + e^*$.
\EndWhile
\EndFor
\end{algorithmic}
\end{algorithm}

\newpage
\section{Pseudo-Code for the Offline Bi-criteria Algorithm}\label{sec:final_pseudo}
In this section, we present the pseudo-code of the main algorithm that we use for the experiments in Section \ref{sec:offline_experiments}. Algorithm \ref{alg:final_robust} works as follows: in an outer loop we obtain an estimate $\gamma$ on the value of the optimal solution $\OPT$ via a binary search. For each guess $\gamma$ we define a set function \(g^\gamma(S):= \frac{1}{k} \sum_{i\in[k]} \min\{f_i(S), \gamma\}\). Then, we run an extended algorithm \texttt{ext-$\mathcal{A}$} (either the standard greedy, threshold greedy, or the adapted version of the stochastic greedy) on $g^\gamma$. If at some point the solution $S$ satisfies $\min_{i\in[k]}f_i(S)\geq (1-\epsilon/2)\gamma$, we stop and update the lower $\lb = \min_{i\in[k]}f_i(S)$, since we find a good candidate. Otherwise, we continue. After finishing iteration $\tau$, we check if we realize the guarantee $g^\gamma(\cup_{j=1}^\tau S_j) \geq (1-\beta^\tau)\cdot\gamma$. If not, then we stop and update the upper bound $\ub = \gamma$, otherwise we continue. Finally, we stop the binary search whenever $\lb$ and $\ub$ are sufficiently close. 

\begin{algorithm}[htbp]
\caption{Pseudo-code to get bi-criteria solutions}\label{alg:final_robust}
\begin{algorithmic}[1]
\small
\renewcommand{\algorithmicrequire}{\textbf{Input:}}
\renewcommand{\algorithmicensure}{\textbf{Output:}}
\Require $\epsilon>0$, monotone submodular functions $\{f_i\}_{i\in[k]}$, partition matroid $P_1,\ldots,P_q$, and subroutine $\mathcal{A}$.
\Ensure sets $S_1,\ldots, S_\ell\in \I$.
\State Compute $\lb$ and $\ub$ as stated above.
\While {$\frac{\ub-\lb}{\ub} > 2\epsilon$}
\State $\gamma = (\ub+\lb)/2$
\For {$\tau = 1,\ldots, \ell$}
\State $S_\tau =\emptyset$.
\State Compute marginals $\rho(e) = g^\gamma(S+e) - g^{\gamma}(S)$ for all $e\in V$.
\If {$\max_e\rho(e) \leq 0$}

\If {$\min_if_i(\cup_{j=1}^\tau S_j)\geq (1-\epsilon)\gamma$}
\State {\bf Update} $\lb = \min_if_i(\cup_{j=1}^\tau S_j)$
\Else
\State {\bf Update} $\ub = \gamma$
\EndIf
\State {\bf Break}

\Else
\State {\bf Obtain} $S_\tau \leftarrow \mathcal{A}(g^\gamma,\cup_{j=1}^{\tau-1}S_j)$
\State 
\If {$g^\gamma(\cup_{j=1}^\tau S_j) < (1-\beta^\tau)\cdot\gamma$}
\State {\bf Update} $\ub = \gamma$.
\State {\bf Break}
\Else
\If {$\min_if_i(\cup_{j=1}^\tau S_j)\geq (1-\epsilon)\gamma$}
\State {\bf Update} $\lb = \min_if_i(\cup_{j=1}^\tau S_j)$
\State {\bf Break}
\Else
\State {\bf Continue}
\EndIf
\EndIf
\EndIf

\EndFor
\EndWhile
\end{algorithmic}
\end{algorithm}